\documentclass[12pt, draftclsnofoot, onecolumn]{IEEEtran}
\usepackage{cite}
\usepackage[cmex10]{amsmath}
\usepackage{amssymb}
\usepackage{amsthm}
\usepackage{mathrsfs}
\usepackage{bm}
\usepackage[mathscr]{eucal}
\usepackage{amssymb,amsmath,amsthm,amsfonts,latexsym}
\usepackage{amsmath,graphicx,bm,xcolor,url}
\usepackage{graphicx}
\usepackage{latexsym}
\usepackage{CJK}
\usepackage{indentfirst}
\usepackage{geometry}
\usepackage{psfrag}
\usepackage{setspace}
\usepackage{algorithmic}
\usepackage{algorithmic, cite}
\usepackage{algorithm}
\usepackage{array}
\usepackage{mdwmath}
\usepackage{mdwtab}
\usepackage{eqparbox}
\usepackage{url}
\usepackage{epstopdf}
\usepackage{epsfig,epsf,psfrag}
\usepackage{fixltx2e}
\usepackage{verbatim}
\usepackage{textcomp}
\usepackage{psfrag} 

\usepackage{caption}
\usepackage{subcaption}

\theoremstyle{plain}
\newtheorem{mythe}{Theorem}
\theoremstyle{remark}
\newtheorem{mylem}{Lemma}
\theoremstyle{plain}

\theoremstyle{remark}
\newtheorem{mypro}{Proposition}
\theoremstyle{plain}

\theoremstyle{remark}
\newtheorem{myrem}{Remark}

\theoremstyle{remark}

\theoremstyle{remark}

\theoremstyle{remark}

\theoremstyle{remark}

\newtheorem{remark}{Remark}

\catcode`~=11 \def\UrlSpecials{\do\~{\kern -.15em\lower .7ex\hbox{~}\kern .04em}} \catcode`~=13

\allowdisplaybreaks[4]

\newcommand{\calC}{\mathcal{C}}
\newcommand{\calD}{\mathcal{D}}

\newcommand{\calN}{\mathcal{N}}

\newcommand{\bB}{\mathbf{B}}

\newcommand{\bD}{\mathbf{D}}
\newcommand{\be}{\mathbf{e}}
\newcommand{\bE}{\mathbf{E}}

\newcommand{\bF}{\mathbf{F}}
\newcommand{\bg}{\mathbf{g}}
\newcommand{\bG}{\mathbf{G}}
\newcommand{\bh}{\mathbf{h}}
\newcommand{\bH}{\mathbf{H}}

\newcommand{\bI}{\mathbf{I}}

\newcommand{\bJ}{\mathbf{J}}

\newcommand{\bl}{\mathbf{l}}

\newcommand{\br}{\mathbf{r}}
\newcommand{\bR}{\mathbf{R}}

\newcommand{\bS}{\mathbf{S}}

\newcommand{\bU}{\mathbf{U}}
\newcommand{\bv}{\mathbf{v}}

\newcommand{\bx}{\mathbf{x}}

\newcommand{\by}{\mathbf{y}}

\newcommand{\bZ}{\mathbf{Z}}

\DeclareMathAlphabet{\mathbsf}{OT1}{cmss}{bx}{n}
\DeclareMathAlphabet{\mathssf}{OT1}{cmss}{m}{sl}

\DeclareSymbolFont{bsfletters}{OT1}{cmss}{bx}{n}
\DeclareSymbolFont{ssfletters}{OT1}{cmss}{m}{n}
\DeclareMathSymbol{\bsfGamma}{0}{bsfletters}{'000}
\DeclareMathSymbol{\ssfGamma}{0}{ssfletters}{'000}
\DeclareMathSymbol{\bsfDelta}{0}{bsfletters}{'001}
\DeclareMathSymbol{\ssfDelta}{0}{ssfletters}{'001}
\DeclareMathSymbol{\bsfTheta}{0}{bsfletters}{'002}
\DeclareMathSymbol{\ssfTheta}{0}{ssfletters}{'002}
\DeclareMathSymbol{\bsfLambda}{0}{bsfletters}{'003}
\DeclareMathSymbol{\ssfLambda}{0}{ssfletters}{'003}
\DeclareMathSymbol{\bsfXi}{0}{bsfletters}{'004}
\DeclareMathSymbol{\ssfXi}{0}{ssfletters}{'004}
\DeclareMathSymbol{\bsfPi}{0}{bsfletters}{'005}
\DeclareMathSymbol{\ssfPi}{0}{ssfletters}{'005}
\DeclareMathSymbol{\bsfSigma}{0}{bsfletters}{'006}
\DeclareMathSymbol{\ssfSigma}{0}{ssfletters}{'006}
\DeclareMathSymbol{\bsfUpsilon}{0}{bsfletters}{'007}
\DeclareMathSymbol{\ssfUpsilon}{0}{ssfletters}{'007}
\DeclareMathSymbol{\bsfPhi}{0}{bsfletters}{'010}
\DeclareMathSymbol{\ssfPhi}{0}{ssfletters}{'010}
\DeclareMathSymbol{\bsfPsi}{0}{bsfletters}{'011}
\DeclareMathSymbol{\ssfPsi}{0}{ssfletters}{'011}
\DeclareMathSymbol{\bsfOmega}{0}{bsfletters}{'012}
\DeclareMathSymbol{\ssfOmega}{0}{ssfletters}{'012}

\newcommand{\balpha}{\bm{\alpha}}

\newcommand{\bomega}{\bm{\omega}}

\newcommand{\bGamma}{\bm{\Gamma}}

\newcommand{\bSigma	}{\bm{\Sigma}}
\newcommand{\bPsi}{\bm{\Psi}}

\newcommand{\bOmega}{\bm{\Omega}}

\newcommand{\bPi}{\bm{\Pi}}
\newcommand{\bTheta}{\bm{\Theta}}

\def\norm#1{\left\| #1 \right\|}
\def\norm2#1{\left\| #1 \right\|_2}
\def\norm22#1{\left\| #1 \right\|_2^2}

\DeclareMathOperator{\diag}{diag}

\DeclareMathOperator{\rank}{rank}

\newcommand{\qednew}{\nobreak \ifvmode \relax \else
      \ifdim\lastskip<1.5em \hskip-\lastskip
      \hskip1.5em plus0em minus0.5em \fi \nobreak
      \vrule height0.75em width0.5em depth0.25em\fi}

\newcommand{\trace}{\mathop{\mathrm{Tr}}}%

\title{Intelligent Reflecting Surface Assisted Non-Orthogonal Multiple Access}
\author{Gang~Yang, \emph{Member, IEEE}, Xinyue Xu, Ying-Chang~Liang, \emph{Fellow, IEEE}
\thanks{G.~Yang and X. Xu are with the National Key Laboratory of Science and Technology on Communications, and the Center for Intelligent Networking and Communications (CINC), University of Electronic Science and Technology of China (UESTC), Chengdu 611731, China (e-mails: yanggang@uestc.edu.cn, 201821220447@std.uestc.edu.cn).}
\thanks{Y.-C. Liang is with the Center for Intelligent Networking and Communications (CINC), University of Electronic Science and Technology of China (UESTC), Chengdu 611731, China (e-mail: liangyc@ieee.org). (\emph{Corresponding author: Y.-C. Liang.})}}

\begin{document}
 \maketitle

\vspace{-0.5cm}
\begin{abstract}
Intelligent reflecting surface (IRS) is a new and disruptive technology to achieve spectrum- and energy-efficient as well as cost-efficient wireless networks. This paper considers an IRS-assisted downlink non-orthogonal-multiple-access (NOMA) system. To optimize the rate performance and ensure user fairness, we maximize the minimum decoding signal-to-interference-plus-noise-ratio (i.e., equivalently the rate) of all users, by jointly optimizing the (active) transmit beamforming at the base station (BS) and the phase shifts (i.e., passive beamforming) at the IRS. A combined-channel-strength based user ordering scheme is first proposed to decouple the user-ordering design and the joint beamforming design. Efficient algorithms are further proposed to solve the formulated non-convex problem for the cases of a single-antenna BS and a multi-antenna BS, respectively, by leveraging the block coordinated decent and semidefinite relaxation (SDR) techniques. For the single-antenna BS case, the optimal solution for the power allocation at the BS and the asymptotically optimal solution for the phase shifts at the IRS are obtained in closed forms. For the multi-antenna BS case, it is shown that the rank of the SDR solution to the transmit beamforming design is upper bounded by two. Also, the convergence proof and the complexity analysis are given for the proposed algorithms. Simulation results show that the IRS-assisted downlink NOMA system can enhance the rate performance significantly, compared to traditional NOMA without IRS and traditional orthogonal multiple access with/without IRS. In addition, numerical results demonstrate that the rate degradation due to the IRS's finite phase resolution is slight, and good rate fairness among users can be always guaranteed.
\vspace{-0.5cm}
\end{abstract}


\section{Introduction}

\subsection{Motivation}

Non-orthogonal multiple access (NOMA) which can serve multiple users in the same resource (e.g., time, frequency, code) block, has been recognized as a promising technology for future wireless communication systems, due to its appealing advantages such as enhanced spectrum efficiency (SE), massive wireless connectivity and low latency \cite{NOMAIN5G}. Specifically, power-domain NOMA exploits the difference in the channel gain among multiple users for multiplexing and relies on successive-interference-cancellation (SIC) for decoding multiple data flows \cite{higuchi2015non}. However, it should be noticed that NOMA achieves significant SE gain than traditional orthogonal multiple access (OMA) only when the channel strengths of multiple users are quite different \cite{DingNOMA}. Moreover, NOMA does not always outperform OMA, for instance, when the channel vectors of users are mutually orthogonal in a downlink multi-input-single-output (MISO) system, the traditional OMA (i.e., spatial division multiple access) is more preferable and no gain can be obtained through traditional NOMA.

On the other hand, intelligent reflecting surface (IRS), also termed as reconfigurable intelligent surfaces (RIS) or large intelligent surface (LIS), is emerged as a new and disruptive technology to achieve spectrum- and energy-efficient as well as cost-efficient wireless networks, thus has drawn fast-growing interests from both academia \cite{LiaskosNieCM2018} and industry~\cite{docomoIRS18}. IRS consists of a large number of low-cost reflecting elements and each element can passively reflect a phase-shifted version of the incident electromagnetic field \cite{CuiIRSLight2014}. The reflected signal propagation can be smartly configured by digitally adjusting the phase shifts of all reflecting elements to achieve certain communication objectives such as received-signal power boosting \cite{LiaskosNieCM2018}, interference mitigation \cite{LiangLISA2019}, and secure transmission \cite{chen2019intelligent}. Hence, IRS can be explored to not only provide additional channel pathes to construct stronger combined channels with significant strength difference, but also re-align the users' (combined) channels to obtain NOMA gain in specific scenarios, in an artificial manner.



IRS provides a new approach to enhance the NOMA performance by reconfiguring the wireless environment, which motivates us to integrate IRS with downlink NOMA system. This is an extension of previous conference-version paper \cite{YangLiangNOMAIRS19}, which extends from a single-input-single-output (SISO) IRS-assisted NOMA system to a MISO IRS-assisted NOMA system. In particular, we aim to maximize the minimum rate of NOMA users by jointly optimizing the transmit beamforming at the BS and the phase shifts at the IRS. This problem is challenging, as the user ordering is determined by the phase shifts at the IRS, due to the fact that the phase shifts affect the combined channel strengths of all users.

There are few parallel works on IRS-assisted NOMA \cite{DingPoorCL19,FuShi19,MuNaofalNOMAIRS19}. In particular, multiple IRSs are employed to maximize the number of served users by effectively reconstructing the users' channels in \cite{DingPoorCL19}. By jointly optimizing the base station's (BS's) transmit beamforming and the IRS's phase shifts, the total power at the BS is minimized in \cite{FuShi19}, and the sum rate was maximized in \cite{MuNaofalNOMAIRS19}. To our best knowledge, there is no existing work focusing on max-min rate optimization for such an IRS-assisted NOMA system with user fairness guarantee.


\subsection{Related works}
\subsubsection{Literature on NOMA}%

NOMA was comprehensively studied in single-carrier communication systems.  First, for single-antenna NOMA systems, the ergodic sum rate and the outage probability performance were analyzed in~\cite{NOMADingSignalProcessing}, while optimal transmission power allocation with user fairness guarantee was investigated in~\cite{FairnessSignal}. Then, for downlink MISO NOMA systems, the beamforming design has been studied in lots of prior works. For example, the authors in~\cite{FzhuBeamNOMA } formulated a weighted sum rate maximization problem with decoding order constraints and quality-of-service constraints for each user. The authors in~\cite{BeamMisoNOMA} formulated an energy efficiency maximization problem with user rate requirements and transmit power constraints. And a novel mathematical programming based approach (named as minorization-maximization algorithm) was provided in~\cite{AMimaOptisum} to solve a sum rate maximization problem. For a multi-input-multi-output (MIMO) NOMA system, a NOMA scheme using intra-beam superposition coding and SIC was proposed in~\cite{higuchi2015non}, in which the basic resource allocation (bandwidth and power) for NOMA was also discussed.
 A power allocation scheme was proposed to maximize the ergodic capacity in~\cite{QsunErgodic}. A novel cluster beamforming strategy was proposed to minimize the total power consumption by jointly optimizing the beamforming and power allocation of a MIMO-NOMA system in~\cite{DingfMIMONOMACluster}.



In addition, NOMA was also studied in multi-carrier communications. The subcarrier and power allocation were jointly optimized to maximize the energy efficiency and the weighted sum throughput for a downlink multi-carrier NOMA network, in~\cite{FangEnergyRe} and~\cite{YsunPowerSubca}, respectively. Also, NOMA was integrated with other communication technologies such as millimeter wave communications \cite{PowBeamNOMAMilli} and backscatter communications \cite{ZhangLiangSRAccess19}. In particular, a backscatter-NOMA system which integrates NOMA and a novel symbiotic radio paradigm (also termed as cooperative ambient backscatter communication~\cite{CooperationAmBC} \cite{LongYangLiangSR18}) was proposed in \cite{ZhangLiangSRAccess19}, and the outage as well as ergotic rate performances were analyzed therein.

\subsubsection{Literature on IRS-assisted Communciations}

IRS which can digitally manipulate the reflected electromagnetic waves is verified to enhance the performance of various wireless communication systems. For an IRS-assisted multiuser MISO wireless system, the transmit beamforming at the BS and the passive beamforming at the IRS were jointly optimized to minimize the total transmission power in \cite{QwuIRS}, while the power allocation at the BS and the phase shifts at the IRS were jointly optimized to maximize the system's sum rate in \cite{HuangYuenIRSICASSP18} and energy efficiency in \cite{HuangYuenIRS18}p, respectively. The minimum-secrecy-rate maximization problem was studied in \cite{ChenLiangIRSAccess2018} for an IRS-assisted downlink multiuser MISO system with multiple eavesdroppers. In\cite{JungLISRateOutage}, the outage probability and the asymptotic distribution of the sum rate was analyzed for an uplink IRS system. In \cite{LIS_quantization}, the ergodic SE performance of an IRS-assisted large-scale antenna system was evaluated and the effect of the phase shifts on the ergodic SE was also investigated. In \cite{TahaCSIRS19}, a deep learning based method was proposed to estimate the IRS's channels, in which the wireless propagation and the IRS are treated as a deep neural network and neurons, respectively.

IRS-assisted communication resembles but differs from existing technologies such as the amplify-and-forward (AF) relay communications and the backscatter communications. A full-duplex AF relay actively processes the received signals before retransmitting the amplified signals, while an IRS passively reflects the signals instead of signal amplification and retransmission, reducing the energy consumption to a large extent \cite{bjornson2019intelligent}. IRS operates in full-duplex mode without introducing self-interference, and provides additional path for the traditional wireless signal without conveying its own information, while a backscatter device transmits its own information by modulating the incident signal from either dedicated source or ambient source \cite{AmbientOFDMCarrier,ZhangLiangJSAC18,LongLiangSRAccess19}.

%

\subsection{Contributions}
In this paper, we consider an IRS-assisted downlink NOMA communication system in which a BS transmits downlink superposed signals to multiple users with the aid of an IRS. The contributions of this paper are outlined as follows:
\begin{itemize}
     \item In order to optimize the rate performance and ensure user fairness, we formulate a problem to maximize the minimum target decoding signal-to-interference-plus-noise-ratio (SINR) (i.e., equivalently the rate) of all users, by jointly optimizing the active beamforming vectors at the BS and the phase shifts (i.e., passive beamforming) at the IRS, subject to the BS's transmit power constraints, the IRS's phase-shift constraints, and the users' SINR constraints for SIC decoding. However, the formulated problem is challenging in two folds: first, the user ordering and the optimization variables are highly coupled; secondly, the problem under given user ordering is still non-convex and difficult to be solved optimally.

\item To decouple the formulated problem, an efficient combined-channel-strength (CCS) based user ordering scheme is proposed for the IRS-assisted NOMA system, which orders the users according to their maximally achievable combined channel strengths by optimizing the IRS's phase shifts. Numerical results show that the designed ordering scheme achieves almost the same rate performance compared to the high-complexity scheme of exhaustively searching over all possible user orderings.

\item To tackle the formulated difficult problem, we first solve it for the special case of a single-antenna BS, in which the beamforming vectors reduces to the power allocation at the BS. We propose an efficient iterative algorithm based on the block coordinated decent (BCD) technique and the semidefinite relaxation (SDR) technique. In each iteration, the optimal power allocation solution for given IRS's phase shifts is derived in closed form, and for the two-user NOMA scenario, the asymptotically optimal phase-shift solution is obtained in closed form. The complexity is analyzed for the proposed algorithm.
\item An extended iterative algorithm is further proposed for the general case of a multi-antenna BS. In each iteration, the SDR technique is applied to obtain a solution of the phase shifts at the IRS for given beamforming matrix at the BS, as well as a solution of the beamforming matrix at the BS for given phase shifts at the IRS. For the transmit beamforming matrix optimization subproblem, it is shown that the rank of the SDR solution is no more than two, independent of the number of NOMA users. The convergence is proved and the complexity is analyzed for the proposed algorithm.
\item Numerical results show that significant rate gains are achieved by our proposed design and algorithms, compared to three benchmark schemes including the traditional NOMA, IRS-assisted OMA and traditional OMA. Moreover, numerical results show that IRS structure with low phase resolution can approximate the best-achievable rate performance of the ideal continuous phase shifts. In addition, numerical results reveal that there is a tradeoff between the number of NOMA users and the achievable sum rate, and good rate fairness can always be guaranteed as the number of NOMA users increases.
\end{itemize}

The rest of this paper is organized as follows. Section~\ref{systemmodel} presents the system model for an IRS-assisted downlink NOMA, and formulates the minimum-throughput maximization problem. Section~\ref{userorderingmul} describes the  CCS-based user ordering scheme. Section~\ref{solutionSin} proposes solutions for the case of a single-antenna BS. Section~\ref{solutionMul} extends the solutions to the general case of a multi-antenna BS. Section~\ref{simulation} presents the numerical results to verify the performance of the proposed joint design. Section~\ref{conclusion} concludes this paper.

The main notations in this paper are listed as follows: The lowercase, boldface lowercase, and boldface uppercase letters, e.g., $g$, $\bg$, and $\bG$, denote a scalar, vector, and matrix, respectively. $|g|$ means the operation of taking the absolute value of a scalar $g$. $\|\bg\|$ means the operation of taking the norm value of a vector $\bg$. $[\bg]^T$ denotes the transpose of a vector $\bg$. $[\bg]^H$ denotes the hermition of a vector $\bg$. $\calC \calN(\mu, \sigma^2)$ denotes the circularly symmetric complex Gaussian (CSCG) distribution with mean $\mu$ and variance $\sigma^2$. $\calC$ denotes the set of complex numbers.

\section{System Model And Problem Formulation}\label{systemmodel}
\subsection{System Model}
As illustrated in Fig. \ref{fig:model}, we consider an IRS-assisted downlink NOMA communication system,  in which a BS equipped with $N \ (N \geq 1)$ antennas transmits superposed signals to $K \ (K\ge 2)$ single-antenna users in the same time and frequency block with the aid of an IRS. The IRS consists of $M \ (M\ge1)$ passive reflecting elements, and each element can reflect a phase-shifted version of the incident signal.  A smart controller which is connected to the IRS can intelligently adjust the phase shifts to assist the NOMA transmission.
\begin{figure} [!t]
	\centering	\includegraphics[width=.7\columnwidth]{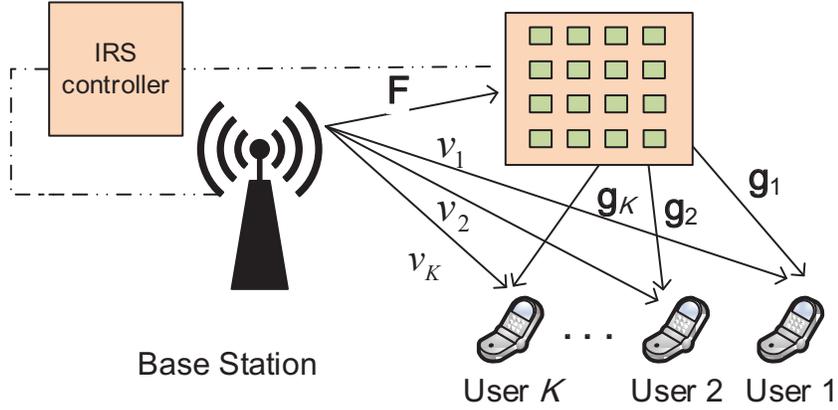}
	\caption{Illustration of an IRS-assisted NOMA System.} \label{fig:model}
\vspace{-0.4cm}
\end{figure}


The channel between the BS and user $j, \ j=1,\ldots,K$, is denoted as $\bv_{j} \in \calC^{N\times1}$. Since the line-of-sight (LoS) path may be blocked, all BS-to-User channels $\bv_j$'s are assumed to be mutually independent and Rayleigh fading. The elements of $\bv_{j}$ are independent and each element follows the distribution $\calC \calN(0, \rho_j)$, where $\rho_j$ denotes the large-scale pathloss from the BS to user $j$. Notice that the IRS is typically pre-deployed such that it can exploit LoS path with the fixed BS. Hence, we use Rician fading to model the channel matrix $\bF \in \calC^{M \times N}$ between the BS and the reflecting elements of the IRS, i.e., 
\begin{align}
\bF =\sqrt{\frac{K_{1} \kappa}{K_{1}+1}}\bar{\bF }+\sqrt{\frac{\kappa}{K_{1}+1}}\tilde{\bF },
\end{align}
where $\kappa$ denotes the large-scale pathloss from the BS to the IRS, $K_{1}$ is the Rician factor of $\bF$, $\bar{\bF }\in \calC^{M\times1}$ and $\tilde{\bF } \in \calC^{M\times1}$ are the LoS component and non-LoS (NLoS) component, respectively. The elements of $\tilde{\bF }$ are independent and each element follows the distribution  $\calC \calN(0,1)$. Since the IRS is typically deployed close to the users to enhance their performance, the channel vector between the IRS and each user $k$ is modeled as
\begin{align}
\bg_j =\sqrt{\frac{K_{2} \beta_j}{K_{2}+1}}\bar{\bg}_j+\sqrt{\frac{\beta_j}{K_{2}+1}}\tilde{\bg}_j,
\end{align}
 where $\beta_j$ denotes the large-scale pathloss from the IRS to the user $j$, $K_{2}$ is the Rician factor of $\bg$, $\bar{\bg}_j \in \calC^{M\times1}$ and $\tilde{\bg}_j \in \calC^{M\times1}$ are the LoS component and the NLoS component, respectively, which are generated similar to $\bar{\bF}$ and $\tilde{\bF}$. We assume that all channel state information is perfectly known.

The BS transmits a superposition of $K$ data flows, each of which is assigned with one dedicated beamforming vector. That is, the transmitted complex baseband signal is
\begin{align}
x=\sum_{j=1}^{K}\bomega_{j} x_{j},
\end{align}
where $x_{j}\sim \calC \calN(0,1)$ is the data flow intended to user $j$ and $\bomega_{j}$ is the corresponding beamforming vector. 
The signal received at user $j$ is then given by
\begin{align}
y_{j}=(\mathbf{g}_{j}^{H}\mathbf{\Theta}\mathbf{F}+\mathbf{v}_{j}^{H})x+n_{j},
\end{align}
where the IRS's diagonal phase-shift matrix $\mathbf{\Theta} \triangleq \diag \{e^{j\theta_{1}},\cdots,e^{j\theta_{M}}\}$ with $\theta_{m}\in[0,2\pi)$ denotes the phase shifts of the reflecting elements, and $n_{j}\sim\mathcal{CN}(0,\sigma^{2})$ denotes the additive white Gaussian noise (AWGN) at  user $j$.

The users in downlink NOMA systems employ SIC technique to decode signals. Preliminary, it is necessary to order users according to their effective channel gains \cite{NOMAIN5G} \cite{NOMADingSignalProcessing}. However, for IRS-assisted NOMA, since the combined channel $\bh_j \triangleq \mathbf{g}_{j}^{H}\mathbf{\Theta}\mathbf{F}+\bv_{j}$ also depends on the phase-shift values $\mathbf{\Theta}$, the optimal user ordering may be any one of all the $K!$ different orders. We will propose an efficient user ordering scheme in Section~\ref{userorderingmul}.

After the user ordering operation, we assume that user $j$ has the $k$-th weakest combined channel, for $k=1,\ldots,K$, thus user $j$ is termed as the $k$-th user in the sequel of this paper (except Section \ref{userorderingmul} for user ordering design), for notational convenience. The $k$-th user is always able to sequentially decode the signals of the $t$-th user, for $t =1, \ldots, k-1$, and then extract them from the received signal. The corresponding SINR for the $k$-th user decoding the $t$-th data flow is given by
\begin{align}\label{gammatk}
\gamma_{t \rightarrow k}&=\frac{|(\mathbf{g}_{k}^{H}\mathbf{\Theta}\mathbf{F}+\mathbf{v}^{H}_{k})\bomega_{t}|^{2}}{\sum \limits_{i=t+1}^{K}|(\mathbf{g}_{k}^{H}\mathbf{\Theta}\mathbf{F}+\mathbf{v}^{H}_{k})\bomega_{i}|^{2}+\sigma^{2} }, \qquad 1\leq t< k \leq K.
\end{align}

After cancelling the interference signals from all weaker users each with index in the set $\{1,\ldots,k-1\}$, the $k$-th user decodes the $k$-th data flow by treating the signals from the rest users as interference. The SINR for the $k$-th user decoding its own signal is expressed as
\begin{align}\label{gammak}
\gamma_{k \rightarrow k}=\frac{|(\mathbf{g}_{k}^{H}\mathbf{\Theta}\mathbf{F}+\mathbf{v}^{H}_{k})\bomega_{k}|^{2}}{\sum \limits_{i=k+1}^{K}|(\mathbf{g}_{k}^{H}\mathbf{\Theta}\mathbf{F}+\mathbf{v}^{H}_{k})\bomega_{i}|^{2}+\sigma^{2} },1\leq k\leq K.
\end{align}

In order to ensure that the $k$-th user can decode the $t$-th data flow correctly and cancel its resulting interference, the SINR for the $k$-th user decoding the $t$-th data flow (i.e., $\gamma_{t\rightarrow k}$) needs to be no smaller than the target SINR of the $t$-th user, denoted by $\gamma_t^{\sf tar}$ \cite{AMimaOptisum}\cite{liu2017joint}, i.e., $\gamma_{t \rightarrow k} \geq \gamma_t^{\sf tar}, \forall t < k$. Then the target SINR of the $t$-th user is expressed as follows
\begin{align}
\gamma_t^{\sf tar}=\min\{\gamma_{t \rightarrow t}, \gamma_{t \rightarrow t+1}, \ldots, \gamma_{t \rightarrow K}\}, \label{gammat_target} \;\; \forall t.
\end{align}

Thus, the corresponding rate for the $t$-th user is thus
\begin{align}\label{Rk}
R_{t}\!=\!\log_2\left(1 \!+\!\gamma_t^{\sf tar} \right),  \forall t.
\end{align}
From \eqref{Rk}, the introducing of IRS is beneficial for downlink NOMA systems, since the use of IRS provides not only additional channel pathes (i.e., $\mathbf{g}_{k}^{H}\mathbf{\Theta}\mathbf{F}$) but also additional design degrees of freedom (i.e., $\bTheta$). By jointly designing the beamforming vectors $\{\bomega_{k}\}$ at the BS and the phase shifts $\bTheta$ at the IRS, the rate performance can be improved.

%


\subsection{Problem Formulation}\label{formulation}

To maximize the system's rate performance while ensuring the fairness among users, as in~\cite{FairnessSignal}~\cite{LiuDingFairnessCL2016}, \textcolor{black}{we maximize the minimum target SINR in \eqref{gammat_target} (equivalently the achievable rate in \eqref{Rk}) of users by jointly optimizing the beamforming vectors $\{\bomega_{k}\}$ at the BS and the phase shifts $\bTheta$ at the IRS.} The optimization problem can be formulated as
\begin{subequations}
\label{eq:P1}
\begin{align}
&\text{(P1): }\underset{\{\bomega_{k}\}, \mathbf{\Theta}, Q}{\max}  \quad  Q \\
&\text{s.t.} \quad \frac{|(\mathbf{g}_{t}^{H}\mathbf{\Theta}\mathbf{F}+\mathbf{v}_{t}^{H})\bomega_{t}|^{2}}{\sum \limits_{i=t+1}^{K}|(\mathbf{g}_{t}^{H}\mathbf{\Theta}\mathbf{F}+\mathbf{v}_{t}^{H})\bomega_{i}|^{2}+\sigma^{2} } \geq Q, \; \forall t \label{eq1:maxminthroughput}\\
&\quad \quad \frac{|(\mathbf{g}_{k}^{H}\mathbf{\Theta}\mathbf{F}+\bv_{k}^{H})\bomega_{t}|^{2}}{\sum \limits_{i=t+1}^{K}|(\mathbf{g}_{k}^{H}\mathbf{\Theta}\mathbf{F}+\bv_{k}^{H})\bomega_{i}|^{2}+\sigma^{2} } \geq Q,  \quad 1\leq t < k \leq K \label{eq1:SICconstraint}\\
&\quad \quad \|\mathbf{g}_{K}^{H}\mathbf{\Theta}\mathbf{F}+\bv_{K}^{H}\|^{2}\geq\|\mathbf{g}_{K-1}^{H}\mathbf{\Theta}\mathbf{F}+\bv_{K-1}^{H}\|^{2} \geq\cdots\geq\|\mathbf{g}_{1}^{H}\mathbf{\Theta}\mathbf{F}+\bv_{1}^{H}\|^{2} \label{eq1:decodingorderconstraint}\\
&\quad \quad \sum \limits_{k=1}^{K}\|\bomega_{k}\|^{2}\leq P \label{eq1:sumpowerallocationconstraint}\\
&\quad \quad 0 \leq\theta_{m} \leq 2\pi,\; \forall m. \label{eq1:Phase-shiftingmatrixconstraint}
\end{align}
\end{subequations}
\textcolor{black}{Note that \eqref{eq1:maxminthroughput} ensures that the SINR of $t$-th user decoding its own signal (i.e., $\gamma_{t \rightarrow t}$) exceeds $Q$, where $Q$ is a slack variable signifying the minimum target SINR to be maximized; \eqref{eq1:SICconstraint} ensures that the SINR of $k$-th user decoding $t$-th data flow (i.e., $\gamma_{t \rightarrow k}$) exceeds $Q$;} \eqref{eq1:decodingorderconstraint} is the combined channel strength constraint under the determined user ordering; \eqref{eq1:sumpowerallocationconstraint} is the constraint of the BS's transmission power; \eqref{eq1:Phase-shiftingmatrixconstraint} is the phase-shift constraints of the IRS's reflecting elements.

Although the optimization problem (P1) is appealing in practice, it is challenging to be solved directly due to the coupled variables (i.e., $\{\bomega_{k}\}, \mathbf{\Theta}, Q$) and the non-convex constraints \eqref{eq1:maxminthroughput} and \eqref{eq1:SICconstraint}. In order to solve Problem (P1) effectively, we will first consider the special case of single-antenna BS in Section \ref{solutionSin}, and then study the general case of multi-antenna BS in Section \ref{solutionMul}.


\section{CCS-based User Ordering Design}\label{userorderingmul}
User ordering is essential for SIC-based decoding and performance optimization in IRS-assisted NOMA systems. With the exhaustive search scheme adopted in \cite{MuNaofalNOMAIRS19}, it needs to solve $K!$ complex sum-rate maximization problems each of which is for particular user ordering, thus the complexity is high especially when the number of users $K$ is relatively large. The BS-to-User channel strength based user ordering adopted in \cite{FuShi19} is simple, but may suffer from poor performance, due to ignoring the effect of the IRS.

In this section, an efficient CCS-based user ordering scheme is designed. That is, the user ordering is determined according to all users' maximally achievable combined channel strengths each of which is obtained by optimizing the IRS phase shifts. Under the designed ordering, it just needs to solve the max-min rate maximization problem once. Numerical results in Section \ref{solution} will show that the designed ordering scheme suffers fromp slight rate degradation compared to exhaustive search scheme.


Specifically, the maximally achievable strength of the combined channel for the $j$-th user can be obtained by solving the following problem
\begin{subequations}
\label{eq:P0}
\begin{align}
&\text{(P-Order): }\underset{ \mathbf{\Theta}}{\max}  \quad  \|\mathbf{g}_{j}^{H}\mathbf{\Theta}\mathbf{F}+\bv_{j}^{H}\|^{2} \\
&\text{s.t.} \quad  0 \leq\theta_{m} \leq 2\pi,\; \forall m. \label{eq0.1:Phase-shiftingmatrixconstraint}
\end{align}
\end{subequations}
Recall $\mathbf{\Theta}=\diag \{e^{j\theta_{1}},\ldots,e^{j\theta_{m}},\ldots,e^{j\theta_{M}}\}$. We denote $e_{m}=e^{j\theta_{m}}$, and $\be=[e_{1}, \ldots, e_{M}]^H$. Then the constraint ~\eqref{eq0.1:Phase-shiftingmatrixconstraint} is equivalent to $|e_{m}|=1,\forall m$. Let  $\bJ_{j}=\diag(\mathbf{g}_{j}^{H})\mathbf{F}$. Then the term $\|\mathbf{g}_{j}^{H}\mathbf{\Theta}\mathbf{F}+\bv_{j}^{H}\|^{2}$ can be rewritten as $\|\be^{H}\bJ_{j}+\bv_{j}^{H}\|^{2}$.
By introducing $\bar{\be}=[\be;1]$ and
\begin{align}\label{eq:Sj}
\bS_{j}=\left[
\begin{array}{ccc}
 \bJ_{j} \bJ_{j}^{H}& \bJ_{j}\bv_{j} \\
\bv_{j}^{H}\bJ_{j}^{H} & 0
\end{array}
\right],
\end{align}
the $\|\be^{H}\bJ_{j}+\bv_{j}^{H}\|^{2}$ can be further transformed to $\bar{\be}^{H}\bS_{j}\bar{\be}+\|\bv_{j}^{H}\|^{2}$.
Note that $\bar{\be}^{H}\bS_{j}\bar{\be}=\trace(\bS_{j}\bar{\be}\bar{\be}^{H})$. Define the matrix $\bE=\bar{\be}\bar{\be}^{H}$, which needs to satisfy $\bE\succeq0$ and $\rank(\bE)=1$. Since the rank-one constraint is non-convex, we exploit the SDR technique to relax (P-Order) as follows
\begin{subequations}
\label{eq:P0-SDR}
\begin{align}
\text{(P-Order-SDR): }&\underset{ \mathbf{E}}{\max}  \quad  \trace(\bS_{j}\mathbf{E}) \!+\! \|\bv_{j}\|^{2} \\
& \text{s.t.} \quad \quad \mathbf{E}\succeq 0 \label{eq2.2:powerallocationconstriant}\\
&\quad \quad \quad  E_{m,m}==1. \label{eq2.2:sumpowerallocationconstriant}
\end{align}
\end{subequations}

The optimal $\mathbf{E}$ obtained by solving (P-Order-SDR) generally does not satisfy the rank-one constraint, and a Gaussian randomization scheme can be applied to obtain a rank-one solution, which is described as follows. Firstly, we obtain the eigenvalue decomposition of $\bE$ as $\bE=\bU\bSigma \bU^{H}$. Define $\bSigma^{\frac{1}{2}}\triangleq \diag\{\sqrt{\lambda_{1}},  \cdots, \sqrt{\lambda_{M+1}}\}$. A random vector is generated as follows
\begin{align}
\tilde{\be}=\bU \bSigma^{\frac{1}{2}}\br,
\end{align}
where the random vector $\br\sim \mathcal{CN}(\mathbf{0},\mathbf{I}_{M+1})$. Then we generate $\hat{\be}$ as $\hat{\be}=e^{j \angle([\frac{\tilde{\be}}{\tilde{\be}_{M+1}}]_{(1:M)})}$, where $[\bx]_{(1:M)}$ denotes the vector containing the first $M$ elements in $\bx$. The objective value of (P-Order) is approximated as the maximal one achieved by the best $\hat{\be}$ among all $\br$'s. It has been verified that SDR technique followed by such randomization scheme can guarantee at least a $\frac{\pi}{4}$-approximation of the optimal objective value of (P-Order) \cite{QwuIRS}.

After solving the above $K$ optimization problems for user ordering, we can obtain the maximally achievable combined channel strength of each user. The user ordering is determined according to the users' achievable combined channel strengths.




%

\section{Optimal Solution for Single-Antenna Base Station Case}\label{solutionSin}
In this section, to obtain more insights on the system design, we solve the max-min rate maximization problem for the single-antenna BS case, i.e., $N=1$. In this case, the channel matrix $\bF \in \calC^{M \times N}$ between the BS and the IRS reduces to the channel vector $\mathbf{f} \in \calC^{M \times 1}$, the channel vector $\bv_k$ between the BS and the $k$-th user reduces to the channel coefficient $v_k$, and
the beamforming vectors $\{\bomega_{k}\}$  in (P1) reduce to the power allocation vector $\balpha R^{K \times N}$. The designed CCS-based user ordering scheme is used herein. Thus Problem (P1) for the single-antenna BS case is rewritten as follows
\begin{subequations}
\label{eq:P2}
\begin{align}
&\text{(P2): }\underset{\mathbf{\alpha} , \mathbf{\Theta}, Q}{\max}  \quad  Q \\
&\text{s.t.} \quad \frac{\alpha_{k} P|\mathbf{g}_{k}^{H}\mathbf{\Theta}\mathbf{f}+v_{k}^{H}|^{2}}{\sum \limits_{i=k+1}^{K}\alpha_{i} P|\mathbf{g}_{k}^{H}\mathbf{\Theta}\mathbf{f}+v_{k}^{H}|^{2}+\sigma^{2} } \geq Q, \; \forall k \label{eq2:maxminthroughput}\\
&\quad \quad |\mathbf{g}_{K}^{H}\mathbf{\Theta}\mathbf{f}+v_{K}^{H}|^{2}\geq|\mathbf{g}_{K-1}^{H}\mathbf{\Theta}\mathbf{f}+v_{K-1}^{H}|^{2}\geq\cdots\geq|\mathbf{g}_{1}^{H}\mathbf{\Theta}\mathbf{f}+v_{1}^{H}|^{2} \label{eq2:decodingorderconstraint}\\
&\quad \quad \sum \limits_{k=1}^{K}\alpha_{k}\leq 1 \label{eq2:sumpowerallocationconstraint}\\
&\quad \quad \alpha_k \geq 0, \; \forall k \label{eq2:powerallocationconstraint}\\
&\quad \quad 0 \leq\theta_{m} \leq 2\pi,\; \forall m. \label{eq2:Phase-shiftingmatrixconstraint}
\end{align}
\end{subequations}

  \textcolor{black}{It can be easily verified that once the user ordering constraint \eqref{eq2:decodingorderconstraint} is satisfied, the inequality $\gamma_{t\rightarrow k}\geq \gamma_{t\rightarrow t}$ always holds. Thus, the constraint $\gamma_{t\rightarrow k}\geq Q$ (similar to \eqref{eq1:SICconstraint}) is omitted in the single-antenna case.} \eqref{eq2:sumpowerallocationconstraint} and \eqref{eq2:powerallocationconstraint} are the normalization constraint and non-negative constraints of the BS's power allocation coefficients.

  It is difficult to solve problem (P2) due to the non-convex constraints \eqref{eq2:maxminthroughput} and \eqref{eq2:decodingorderconstraint}. Since the two blocks of variables $\balpha$ and $\bTheta$ are coupled in \eqref{eq:P2}, we exploit the BCD (i.e., blocked coordinate decent) and SDR (i.e., semidefinite-relaxation) techniques to solve it approximately. In each iteration $n=1,2,\ldots$, we optimize different blocks of variables alternatively. Therefore, Problem (P2) is decoupled into two subproblems in each iteration, as described in Subsection \ref{PhaseShiftOptimizationSin} and Subsection \ref{Power Allocation OptimizationSin}. In the sequel, $\balpha^n$ and $\bTheta^n$ with superscript $n$ indicate their values after the $n$-th algorithmic iteration.

\subsection{Phase Shift Optimization}\label{PhaseShiftOptimizationSin}
\subsubsection{General Solution for $K \geq 2$}
In each iteration $n$, for given power allocation coefficients $\balpha^{n-1}$, the phase shifts $\bTheta$ can be optimized by solving the following problem
\begin{subequations}
\label{eq:P2.1}
\begin{align}
&\text{(P2.1):}\quad \underset{\bTheta,Q}{\max}  \quad  Q \\
&\text{s.t.} \quad \eqref{eq2:maxminthroughput}, \eqref{eq2:decodingorderconstraint}, \eqref{eq2:Phase-shiftingmatrixconstraint}.
\end{align}
\end{subequations}

Let $\bl_{k}=\diag(\mathbf{g}_{k}^{H})\mathbf{f}_{k}$, then the term $|\mathbf{g}_{k}^{H}\mathbf{\Theta}\mathbf{f}+v_{k}|^{2}$ in \eqref{eq2:maxminthroughput} and \eqref{eq2:decodingorderconstraint} can be rewritten as $|\be^{H}\bl_{k}+v_{k}|^{2}$, with ${\be}=[e^{j\theta_{1}},\ldots,e^{j\theta_{m}},\ldots,e^{j\theta_{M}}]^H$. Recall that $\bE=\bar{\be}\bar{\be}^{H}$, with $\bar{\be}=[\be;1]$. By introducing
\begin{align}
\bR_{k}=\left[
\begin{array}{ccc}
 \bl_{k} \bl_{k}^{H}& \bl_{k}v_{k} \\
v_{k}^{H}\bl_{k}^{H} & 0
\end{array}
\right],
\end{align}
we have $|\mathbf{g}_{k}^{H}\mathbf{\Theta}\mathbf{f}+v_{k}|^{2}=\trace(\bR_{k}\mathbf{E}) \!+\! |v_{k}|^{2}$. Problem (P2.1) can be further transformed into the following problem
\begin{subequations}
\label{eq:P2.2}
\begin{align}
&\text{(P2.2):}\quad \underset{\mathbf{E},Q}{\max}  \quad  Q \\
&\text{s.t.}  \; P \left( \!\alpha_{k}^n\!-\!\sum \limits_{i=k+1}^{K}\alpha_{i}^n Q \!\right) \left(\trace(\bR_{k}\mathbf{E}) \!+\! |v_{k}|^{2}\right) \geq Q \sigma^{2}, \forall k \label{eq2.2:maxminthroughput}\\
&\quad \quad \trace(\bR_{K}\mathbf{E}) \!+\! |v_{K}|^{2} \geq \trace(\bR_{K-1}\mathbf{E}) \!+\! |v_{K-1}|^{2} \geq\cdots\geq \trace(\bR_{1}\mathbf{E}) \!+\! |v_{1}|^{2} \label{eq2.2:decodingorderconstraint}\\
&\quad \quad \mathbf{E}\succeq 0 \label{eq2.2:powerallocationconstriant}\\
&\quad \quad  [\bE]_{m,m}=1. \label{eq2.2:sumpowerallocationconstriant}
\end{align}
\end{subequations}

However, (P2.2) is still non-convex due to the non-convex constraint \eqref{eq2.2:maxminthroughput}. It can be solved by using the following two-step procedure. First, the bisection search method can be used to decouple $Q$ and $\bTheta$, and the SDR technique can be used to obtain the optimal $\bE$. Specifically, with certain $Q_{\max}$ and $Q_{\min}$, we replace $Q$ in (P2.2) by $\frac{Q_{\max}+Q_{\min}}{2}$, and solve the resulting feasibility problem reduced from (P2.2). The update of $Q_{\max}$ and $Q_{\min}$ depends on whether a feasible $\bE$ can be found. It can be checked that for sufficiently large $Q_{\max}$ and small $Q_{\min}$, the above bisection search over $Q$ can give a globally optimal phase-shift-related matrix $\bE^{n+1}$ in the $n$-th alterative iteration. Second, the Gaussian randomization technique described in Section \ref{userorderingmul} can be applied to obtain a suboptimal rank-one solution $\be$.

\subsubsection{Closed-form Solution for $K=2$}
The above general procedure for optimizing $\bTheta$ involves the bisection search, SDR and randomization-based approximation, is thus complicated in practice. In this section, a closed-form solution is proposed for the two-user NOMA scenario, i.e., $K=2$.

For notational convenience, we rewrite the channel elements as $[\mathbf{g}_{k}^{H}]_i=|[\mathbf{g}_{k}^{H}]_i| e^{j\varphi_{k,i}}$, $[\mathbf{f}]_{i}=|[\mathbf{f}]_{i}| e^{j\psi_i}$ and $v_{k}= |v_{k}| e^{j\xi_k}$, for $i=1,\ldots, M$, where $[\mathbf{g}_{k}^{H}]_i$ and $[\mathbf{f}]_{i}$ denote the channel from the $k$-th user to the $i$-th reflecting element of the IRS and the channel from the BS to the $i$-th reflecting element of the IRS, respectively. We have the following proposition on the optimal phase shifts $\bTheta$.
\begin{mypro}\label{Proposition:Prop1}
For two-user NOMA scenario, as the transmission power $P$ at the BS increases, the asymptotically optimal phase shifts of the IRS is given by
\begin{align}
\theta_{i}=\xi_2-\varphi_{2,i}-\psi_{i},  \quad {\text{for}} \  i=1,\ldots, M.
\end{align}
\end{mypro}
\begin{proof}
  See proof in Appendix \ref{app:prop1}.
\end{proof}

\begin{myrem}
Notice that the closed-form solution in Proposition 1 is asymptotically optimal for large $P$, but suboptimal for small or moderate $P$. This is explained as follows. When $P$ is small or moderate, the noise term $\sigma^2$ in the denominator of \eqref{gamma11} is not negligible compared to the interference term $\alpha_2 P |h_1|^2$. And both SINRs $\gamma_{1 \rightarrow 1}$ and $\gamma_{2 \rightarrow 2}$ are monotonically increasing functions of $|h_1|^2$ and $|h_2|^2$, respectively. Hence, to maximize the minimum value between both SINRs, the optimal phase shifts $\bTheta$ should be designed to enhance both $|h_1|^2$ and $|h_2|^2$ in general.  Numerical results show that even for small or moderate $P$, the above closed-form solution suffers from slight rate performance degradation compared to the solution achieved by the general algorithm.
\end{myrem}


\subsection{Power Allocation Optimization}\label{Power Allocation OptimizationSin}

In each iteration $n$, for given phase shifts $\bTheta^n$, the power allocation $\balpha$ can be optimized by solving the problem 
\begin{subequations}
\label{eq:P2.3}
\begin{align}
&\text{(P2.3):}\quad \underset{\balpha,Q}{\max} \quad  Q \\
&\text{s.t.} \quad \eqref{eq2:maxminthroughput}, \eqref{eq2:sumpowerallocationconstraint}, \eqref{eq2:powerallocationconstraint}.
\end{align}
\end{subequations}



 We have the following lemma on the optimal solution to (P2.3)
\begin{mylem}
\label{lemma:Lem1}
The optimal power allocation vector of Problem (P2.3), denoted by $\tilde{\balpha}=[\tilde{\alpha}_1, \tilde{\alpha}_2, \ldots, \tilde{\alpha}_K]$, is the unique solution to the following equations
\begin{align}\label{eq:P2.4_equal}
\frac{\tilde{\alpha}_{1} P|h_{1}|^{2}}{\sum \limits_{i=2}^{K}\tilde{\alpha}_{i} P|h_{k}|^{2}+\sigma^{2}}=\cdots=\frac{\tilde{\alpha}_{k} P|h_{k}|^{2}}{\sum \limits_{i=k+1}^{K}\tilde{\alpha}_{i} P|h_{k}|^{2}+\sigma^{2}}=\cdots = \frac{\tilde{\alpha}_{K} P|h_{K}|^{2}}{\sigma^{2}}=Q,
\end{align}
\begin{align}\label{eq:P2.4_sum}
\sum \limits_{k=1}^{K}\tilde{\alpha}_{k}= 1.
\end{align}
\end{mylem}

\begin{proof}
See proof in Appendix \ref{app:lemma1}.
\end{proof}

From Lemma \ref{lemma:Lem1}, the optimal solution $\tilde{\balpha}$ can be obtained by solving $K$ nonlinear equations in \eqref{eq:P2.4_equal} and \eqref{eq:P2.4_sum}. Fortunately, the optimal solution $\tilde{\balpha}$ can be obtained in closed form. To such end, we define the length-$K$ column vector with all-one elements as $\textbf{1}$, the diagonal matrix $\bD=\diag{\{\frac{1}{|h_1|^2}, \frac{1}{|h_2|^2}, \ldots, \frac{1}{|h_K|^2}}\}$, the upper-triangle matrix $\bPsi$ with elements $[\bPsi]_{ik}=|h_i|^2$ for $i<k$ and $[\bPsi]_{ik}=0$ for $i \geq k$, and the following $(K+1)$-dimensional square matrix
\begin{align}
\bPi=\left[\begin{array}{ccc}
 \bD \bPsi & \frac{\sigma^2}{P}\bD  \textbf{1}\\
 \textbf{1}^T  \bD \bPsi  & \frac{\sigma^2}{P}\textbf{1}^T\bD \textbf{1}
\end{array}
\right].
\end{align}
The closed-form solution is then given in the following Theorem \ref{mythm:thm1}.
\begin{mythe}
\label{mythm:thm1}
The optimal objective value of Problem (P2.3) is
\begin{align}
Q=\frac{1}{\lambda_{\max}(\bPi)},
\end{align}
where $\lambda_{\max}(\bPi)$ denotes the largest eigenvalue of the matrix $\bPi$, and the optimal power allocation vector $\tilde{\balpha}$ is obtained as the first $K$ components of $\bPi$'s dominant eigenvector scaled by its last component.
\end{mythe}

\begin{proof}
See proof in Appendix \ref{app:theorem1}.
\end{proof}

\begin{remark}
Recall $h_k\triangleq \mathbf{g}_{K}^{H}\mathbf{\Theta}\mathbf{f}+v_{K}^{H}$, for given $\bTheta$, the optimal solution to Problem (P2.3) satisfies \eqref{eq:P2.4_equal}, thus the SINR of each user will be balanced at the same value at the optimal $\balpha$ and $\bTheta$. That is, for the case of single-antenna BS, the best rate fairness among users can be guaranteed.
\end{remark}

\subsection{Overall Algorithm}
The iterative algorithm for solving problem (P2) is summarized in Algorithm \ref{AlgorithmP1}, in each iteration of which the $\balpha$ and $\bTheta$ are alternatively optimized. The bisection search and the SDR technique are utilized to optimize $\bTheta$ from step 3 to step 23, and Theorem \ref{mythm:thm1} is used to compute the optimal $\balpha$ in step 24. The algorithm ends when the increase of the objective value is smaller than a smaller $\epsilon >0$. It can be shown that Algorithm \ref{AlgorithmP1} is guaranteed to converge, whose proof is similar to that of Algorithm \ref{AlgorithmP2} in Section \ref{solutionMul} and thus omitted herein.
\begin{algorithm}[t!]
\caption{Iterative algorithm for solving problem (P2)}\label{AlgorithmP1}
\begin{algorithmic}[1]
\STATE Initialize $\balpha^0$, $\bTheta^0$, $D$ (a large positive integer) and $\epsilon_{\sf b}$ (a smaller positive value). Let $n=0$.
\REPEAT
\STATE Given $Q_{\max}$, $Q_{\min}$, \WHILE{$Q_{\max}-Q_{\min}\geq\epsilon_{\sf b}$,}
\STATE  Solve the feasibility problem reduced from (P2.2) with given $Q=\frac{Q_{\max}+Q_{\min}}{2}$.
\IF {the feasibility problem reduced from (P2.2) is solvable,}
\STATE$Q_{\min}=Q$, update $\bE$.
\ELSE
\STATE   $Q_{\max}=Q$.
\ENDIF
\ENDWHILE
\RETURN $\bE^{n+1}=\bE$.
\STATE Compute the eigenvalue decomposition of $\bE^{n+1}$ as $\bE^{n+1}=\bU \bSigma \bU^{H}$.
\STATE Initialize $\calD=\emptyset$.
\FOR{$d=1,\ldots,D$}
\STATE Generate a random vector $\tilde{\be}_{d} =\bU \bSigma^{\frac{1}{2}}\br_{d}$, where $\br_{d} \sim \calC \calN(\bold{0}_{M+1}, \bI_{M+1})$.
\STATE Compute $\hat{\be}_{d}=e^{j \angle([\frac{\tilde{\be}_{d}}{\tilde{\be}_{d,M+1}}]_{(1:M)})}$, and then obtain $\bTheta_{d}=\diag(\hat{\be}_{d}^{H})$.
\IF {$\bTheta_{d}$ is feasible for problem (P2.1),}
\STATE $\calD=\calD \bigcup d$
\STATE Obtain the objective value of  (P2.1) as $Q_d$.
\ENDIF
\ENDFOR
\RETURN $\bTheta^{n+1} = \arg \underset{ d \in \calD} {\max} \; Q_d$.
\STATE For given $\bTheta^{n+1}$, use Theorem \ref{mythm:thm1} to compute the optimal power allocation $\balpha^{n+1}$.
\STATE  Update iteration index $n=n+1$.
\UNTIL{The increase of the objective value is smaller than $\epsilon$}.
\STATE  Return the optimal solution $\alpha^{\star}=\alpha^{n}$, $\Theta^{\star}= \Theta^{n}$, and $Q^{\star}$.
\end{algorithmic}
\end{algorithm}


In Algorithm \ref{AlgorithmP1}, the subproblems (P2.1) and (P2.3) are alteratively solved in each outer-layer BCD iteration, and the subproblem (P2.1) is solved by a bisection search inner-layer iteration in each of which an SDR problem (P2.2) is solved. Specifically, it needs to execute $\log_2(\frac{Q_{\sf max}-Q_{\sf min}}{\epsilon_b})$ iterations to achieve an accuracy $\epsilon_{\sf b}$ of bisection search over $Q$. Hence, from the complexity analysis of typical interior-point method like primal-dual path following
method~\cite{ma2010semidefinite}, the complexity of Algorithm \ref{AlgorithmP1} is obtained as $\mathcal{O}\big(I_{\sf ite}\log_2(\frac{Q_{\sf max}-Q_{\sf min}}{\epsilon_b})\max\{2K-1,M+1\}^4\left( M+1\right)^{\frac{1}{2}}\log{\frac{1}{\epsilon_{\sf s}}}\big)$, where $I_{\sf ite}$ denotes the number of outer-layer BCD iterations, and $\epsilon_{\sf s}$ denotes the predefined accuracy of the SDR solution.

\section{Optimal Solution For Multi-Antenna Base Station Case}\label{solutionMul}
In this section, we consider the general case of a multi-antenna BS. The algorithm for the special case of a single-antenna BS is generalized to solve problem (P1). The original problem is decoupled into two subproblems, which are described as the following subsections.

\subsection{ Phase Shift Optimization}\label{PhaseShiftOptimizationMul}
In each iteration $n$, for given beamforming vectors $\{\bomega_{k}\}$, the phase shifts $\bTheta$ can be optimized by solving the following problem
\begin{subequations}
\label{eq:P1.1}
\begin{align}
&\text{(P1.1):}\quad \underset{\bTheta,Q}{\max}  \quad  Q \\
&\text{s.t.} \quad \eqref{eq1:maxminthroughput}, \eqref{eq1:SICconstraint}, \eqref{eq1:decodingorderconstraint}, \eqref{eq1:Phase-shiftingmatrixconstraint}.
\end{align}
\end{subequations}

Recall $\mathbf{\Theta}=\diag \{e^{j\theta_{1}},\ldots,e^{j\theta_{M}}\}$, with $m=1,\ldots,M$. We denote $e_{m}=e^{j\theta_{m}}$, and $\be=[e_{1}, \ldots, e_{M}]^H$. By applying  $\bl_{k,t}=\diag(\mathbf{g}_{k}^{H})\mathbf{F}\bomega_{t}$ and $v_{k,t}=\bv_{k}^{H}\bomega_{t}$, the term $|(\mathbf{g}_{k}^{H}\mathbf{\Theta}\mathbf{F}+\bv_{k}^{H})\bomega_{t}|^{2}$ can be rewritten as $|\be^{H}\bl_{k,t}+v_{k,t}|^{2}$.
We further reduce this term to $\bar{\be}^{H}\bR_{k,t}\bar{\be}+|v_{k,t}|^{2}$, where $\bar{\be}=[\be;1]^H$ and
\begin{align}
\bR_{k,t}=\left[
\begin{array}{ccc}
 \bl_{k,t} \bl_{k,t}^{H}& \bl_{k,t}v_{k,t}^{H} \\
\bl_{k,t}^{H}v_{k,t} & 0
\end{array}
\right]
\end{align}



Following similar manipulations as in Section \ref{userorderingmul}, the term $\|\mathbf{g}_{k}^{H}\mathbf{\Theta}\mathbf{F}+\bv_{k}^{H}\|^{2}$ can be rewritten as $\bar{\be}^{H}\bS_{k}\bar{\be}+\|\bv_{k}^{H}\|^{2}$, where $\bS_k$ is given in \eqref{eq:Sj} with subscript replaced by k.

Note that $\bar{\be}^{H} \bR_{k,t}\bar{\be}=\trace (\bR_{k,t}\bar{\be}\bar{\be}^{H})$, $\bar{\be}^{H}\bS_{k}\bar{\be}=\trace(\bS_{k}\bar{\be}\bar{\be}^{H})$. We define the matrix $\bE=\bar{\be}\bar{\be}^{H}$, which needs to satisfy $\bE\succeq0$ and $\rank(\bE)=1$. Since the rank-one constraint is non-convex, we exploit the SDR technique to relax problem (P1.1) as follows
\begin{subequations}
\label{eq:P1.2}
\begin{align}
&\text{(P1.2):}\quad \underset{\mathbf{E},Q}{\max}  \quad  Q \\
&\text{s.t.}  \frac{\trace(\bR_{k,k}\mathbf{E}) \!+\! |v_{k,k}|^{2}}{\sum \limits_{i=k+1}^{K}(\trace(\bR_{k,i}\mathbf{E}) \!+\! |v_{k,i}|^{2})+\sigma^{2}} \geq Q , \forall k \label{eq1.2:maxminthroughput}\\
&\quad \quad \frac{\trace(\bR_{k,t}\mathbf{E}) \!+\! |v_{k,t}|^{2}}{\sum \limits_{i=t+1}^{K}(\trace(\bR_{k,i}\mathbf{E}) \!+\! |v_{k,i}|^{2})+\sigma^{2}} \geq Q, 1\leq t <k \leq K \label{eq1.2:SICconstriant}\\
&\quad \quad \trace(\bS_{K}\mathbf{E}) \!+\! \|\bv_{K}\|^{2} \geq \trace(\bS_{K-1}\mathbf{E}) \!+\! \|\bv_{K-1}\|^{2}  \geq\cdots\geq \trace(\bS_{1}\mathbf{E}) \!+\! \|\bv_{1}\|^{2}\label{eq2.2:decodingorderconstraint}\\
&\quad \quad \mathbf{E}\succeq 0 \label{eq1.2:powerallocationconstriant}\\
&\quad \quad  [\bE]_{m,m}=1. \label{eq1.2:sumpowerallocationconstriant}
\end{align}
\end{subequations}

Problem (P1.2) is still non-convex due to the non-convex constraint of \eqref{eq1.2:maxminthroughput} and \eqref{eq1.2:SICconstriant}. To tackle the coupled variables $Q$ and $\bTheta$, we use the bisection search method similar to \ref{PhaseShiftOptimizationSin}. Then the Gaussian randomization technique is applied to obtain a approximate solution.

\subsection{Beamforming Matrix Optimization}

In each iteration $n$, for given phase shifts $\bTheta^n$, the beamforming vectors $\{\bomega_{k}\}$ can be optimized by solving the problem 
\begin{subequations}
\label{eq:P1.3}
\begin{align}
&\text{(P1.3):}\quad \underset{\{\bomega_{k}\},Q}{\max} \quad  Q \\
&\text{s.t.} \quad \eqref{eq1:maxminthroughput}, \eqref{eq1:SICconstraint}, \eqref{eq1:sumpowerallocationconstraint}.
\end{align}
\end{subequations}

As (P1.3) is non-convex due to the non-convex constraints \eqref{eq1:maxminthroughput} and p\eqref{eq1:SICconstraint}, we adopt the SDR technique to obtain an efficient approximate solution. Denote the combined channel $\bh_{k}^{H}=\mathbf{g}_{k}^{H}\mathbf{\Theta}\mathbf{F}+\mathbf{v}_{k}^{H}$. By introducing $\bH_{k}=\bh_{k}\bh_{k}^{H}$ and $\bOmega_{k}=\bomega_{k}\bomega_{k}^{H}$, the term $|(\mathbf{g}_{k}^{H}\mathbf{\Theta}\mathbf{F}+\mathbf{v}_{k}^{H})\bomega_{k}|^{2}$ can be written as $\trace{(\bH_{k}\bOmega_{k})}$, where needs to satisfy $\bOmega_{k}\succeq0$ and $\rank{(\bOmega_{k})}=1$. Since the constraint $\rank{(\bOmega_{k})}=1$ is non-convex, we relax this constraint and problem (P1.3)
can be transformed as follows,
\begin{subequations}
\label{eq:P1.4}
\begin{align}
&\text{(P1.4):}\quad \underset{\{\bOmega_{k}\},Q}{\max}  \quad  Q \\
&\text{s.t.}  \frac{\trace(\bH_{k}\bOmega_{k})}{\sum \limits_{i=k+1}^{K}\trace(\bH_{k}\bOmega_{i}) \!+\sigma^{2}} \geq Q , \forall k \label{eq1.4:maxminthroughput}\\
&\quad \quad \frac{\trace(\bH_{k}\bOmega_{t})}{\sum \limits_{i=t+1}^{K}\trace(\bH_{k}\bOmega_{i}) \!+\sigma^{2}} \geq Q , 1\leq t <k \leq K \label{eq1.4:SICconstriant}\\
&\quad \quad \bOmega_{k}\succeq 0 , \forall k\label{eq1.4:powerallocationconstriant}\\
&\quad \quad \sum\limits_{k=1}^{K}\trace(\bOmega_{k})\leq P . \label{eq1.4:sumpowerallocationconstriant}
\end{align}
\end{subequations}
Since problem (P1.4) is non-convex, the method based on the bisection search and SDR technique similar to which described in \ref{PhaseShiftOptimizationSin} can be applied to solve the problem. Furthermore, we have the following theorem on the SDR solution to (P1.4).
\begin{mythe}
\label{mythm:thmSDR}
The rank of the SDR solution \{$\bOmega_k^{\star}$\} to (P1.4) with given $Q$ is upper bounded by two.
\end{mythe}

\begin{proof}
See proof in Appendix \ref{app:theoremSDR}.
\end{proof}

When the SDR solution \{$\bOmega_k^{\star}$\} is rank-one, the optimal transmit beamforming vector can be obtained through Cholesky decomposition as $\bOmega_k^{\star}=\bomega_k^{\star} (\bomega_k^{\star})^H, \ \forall k$. If the SDR solution is rank-one, the randomization-based technique can be similarly adopted as in Section \ref{userorderingmul} to obtain a rank-one solution. Specifically, we firstly obtain the eigenvalue decomposition of $\bOmega_{k}^{n+1}$ as $\bOmega_{k}^{n+1}=\bU_{k}\bSigma_{k} \bU_{k}^{H}$, $k=1,\ldots,K$. Then $K$ random vectors is generated as follows
\begin{align}
\tilde{\bomega}_{k}=\bU_{k} \bSigma_{k}^{\frac{1}{2}}\br_{k},
\end{align}
where the random vector $\br_{k}\sim \mathcal{CN}(\mathbf{0},\mathbf{I}_{M+1})$. The objective value of problem (P1.4) is approximated as the maximum one achieved by the best $\{\tilde{\bomega_{k}}\}$.

\subsection{ Overall algorithm }\label{overallalgorithmP2P1}
The overall algorithm is summarized in Algorithm \ref{AlgorithmP2}. As shown, the algorithm optimizes $\{\bomega_{k}\}$ and $\bTheta$ alternatively in the out-layer iteration, and the SDR technique is adopted to obtain an approximate phase-shift solution and beamforming solution. The algorithm ends when the the incremental increase of the objective value is sufficiently small.
\begin{algorithm}[t!]
\caption{Iterative algorithm for solving problem (P1)}\label{AlgorithmP2}
\begin{algorithmic}[1]
\STATE Initialize \{$\bomega_k^0$\}, $\bTheta^0$, $D$ (a large positive integer) and $\epsilon$ (a smaller positive value). Let $n=0$.
\REPEAT
\STATE For given \{$\bomega_k^n$\}, obtain $\bTheta^{n+1}$ through similar procedure as described in steps 3 to 23 of Algorithm \ref{AlgorithmP1} with Problem (P2.2) replaced by Problem (P1.2).

\STATE For given $\bTheta^{n+1}$, apply the bisection search and the SDR techniques to solve Problem (P1.4), and obtain the optimal SDR solution as $\{\bOmega_{k}^{n+1}\}$, then apply Gaussian-randomization technique to obtain an approximate rank-one solution \{$\bomega_{k}^{n+1}$\}.

\STATE  Update iteration index $n=n+1$.
\UNTIL{The increase of the objective value is smaller than $\epsilon$}.
\STATE  Return the optimal solution $\bomega^{\star}$'s $=\bomega^{n}$'s , $\Theta^{\star}= \Theta^{n}$, and $Q^{\star}$.
\end{algorithmic}
\end{algorithm}

We prove the convergence of Algorithm \ref{AlgorithmP2} as follows.
\begin{mythe}
Algorithm \ref{AlgorithmP2} is guaranteed to converge.
\end{mythe}

\begin{proof}
First, in step 3 of Algorithm \ref{AlgorithmP2}, since the optimal solution $\bTheta^{n+1}$ is obtained for given \{$\bomega_k^{n}$\}, we have the following inequality on the minimum rate 
\begin{align}
    Q(\bTheta^{n},\bomega_k^{n}) \leq Q(\bTheta^{n+1},\bomega_k^{n}). \label{eq:Qinequality1}
\end{align}

Second, in step 4 of Algorithm \ref{AlgorithmP2}, since \{$\bomega^{n}$\} is the optimal solution to Problem (P1.3), the following inequality holds
\begin{align}
  Q(\bTheta^{n+1},\bomega_k^{n}) \leq Q(\bTheta^{n+1},\bomega_k^{n+1}). \label{eq:Qinequality2}
\end{align}

From \eqref{eq:Qinequality1} and \eqref{eq:Qinequality2}, we further have
\begin{align}\label{eq:Qinequality32}
  Q(\bTheta^{n},\bomega_k^{n}) \leq Q(\bTheta^{n+1},\bomega_k^{n+1}).
\end{align}
The inequality in \eqref{eq:Qinequality32} indicates that the objective value of problem (P1) is always non-decreasing after each iteration. On the other hand, since the objective is continuous over the compact feasible set of problem (P1), it is upper-bounded by some finite positive number \cite{CVXBoyd04}. Hence, the proposed Algorithm \ref{AlgorithmP2} is guaranteed to converge, which completes the proof. 
\end{proof}

Notice that no global optimality can be assured for Algorithm \ref{AlgorithmP2}. The reasons are two fold. First, the problem (P1) is not jointly convex with respect to $\bTheta$, \{$\bomega_k$\} and $Q$. Second, the adopted method of SDR followed by Gaussian randomization for solving  sub-problem (P1.2) and (P1.4) does not guarantee the global optimality of solution.


In Algorithm \ref{AlgorithmP2}, the subproblems (P1.2) and (P1.4) are alteratively solved in each outer-layer BCD iteration, and each subproblem is solved by a bisection search inner-layer iteration in each of which an SDR problem is solved. Similar to the complexity analysis for Algorithm \ref{AlgorithmP1}, the complexity of  Algorithm \ref{AlgorithmP2} is obtained from~\cite{ma2010semidefinite} as $\mathcal{O}\big(I_{\sf ite}\big(\log_2(\frac{Q_{\max}-Q_{\min}}{\epsilon_{\sf b}})\max\{\frac{K^2+3K-2}{2},M+1\}^4\left( M+1\right)^{\frac{1}{2}}\log{\frac{1}{\epsilon_{\sf s}}}+\log_2\left(\frac{Q_{\max}-Q_{\min}}{\epsilon_{\sf b}}\right) \max\{\frac{(K+1)K}{2},N\}^4 N^{\frac{1}{2}}\log{\frac{1}{\epsilon_{\sf s}}}\big)\big)$.

\section{NUMERICAL RESULTS}\label{solution}\label{simulation}
Numerical results are provided in this section. Under a three-dimensional (3D) coordinate ($x,y,z$) system, the horizontal projection $(x,y)$ of which is illustrated in Fig. \ref{fig:FigSimlocation}. Assume that both the BS and the IRS are located at the altitude of 10 meter (m), and the locations of the BS and the IRS are set as (0, 0, 10) and ($25\sqrt{2}$, $25\sqrt{2}$, 10), respectively; while the users, at the altitude of 1.5m, are randomly and uniformly distributed in the rectangular area of Fig. \ref{fig:FigSimlocation}. We assume that the BS is equipped with a uniform linear array (ULA) with antenna spacing $d_{\sf B}=\lambda/2$, where $\lambda$ denotes the wavelength; while the IRS is equipped with a uniform rectangular array (URA) with IRS element spacing $d_{\sf I}=\lambda/8$. The operating frequency is assume as $2.5$ GHz. The LoS components $\bar{\bF}$ and $\bg_k$ are modeled by the steering vectors depending on the angle of arrival and angle of departure of particular LoS pathes \cite{LIS_quantization}.


\begin{figure}[!t]
	\centering	\includegraphics[width=.7\columnwidth]{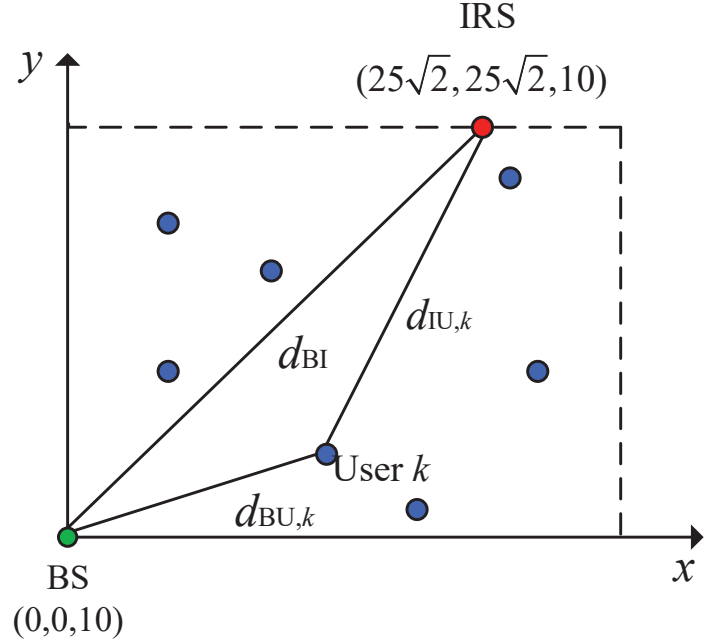}
	\caption{Horizontal locations of an IRS-assisted NOMA system.} \label{fig:FigSimlocation}
\vspace{-0.4cm}
\end{figure}
\vspace{-0.1cm}

As in \cite{QwuIRS}, we assume that the BS-to-User channels are Rayleigh fading and the large-scale pathloss is $10^{-3}d^{-4}$, where $d$ is the distance with unit of meter. Both the BS-to-IRS channel and the IRS-to-User channels are assumed to be Rician fading, and their pathloss are $10^{-3}d^{-2}$ and $10^{-3}d^{-2.5}$, respectively. We set the Rician factors $K_{1}=K_{2}=10$. As in \cite{ZhangLiangLiJSAC} and \cite{ULNOMAYuanCL18}, we set $\sigma^{2}=-114$ dBm. Let $D=400$ and $\epsilon=0.01$.

For communication performance comparison, we consider three benchmarks, i.e., traditional NOMA (without IRS but with optimized power allocation at the BS), IRS-assisted OMA (with optimized power allocation, phase shift and domain of freedom), and the traditional OMA  (without IRS but with optimized power allocation and domain of freedom). The achievable rate of each user for OMA is expressed as (6) in \cite{OptiNOMAOMA}. Simulation results are based on 1000 random channel realizations.


\subsection{Rate Performance Evaluation for Proposed Solution}\label{resultforsrateperformance}
In this subsection, the numerical results for rate performance of the proposed IRS-assisted NOMA and three benchmarks are analyzed. We consider the case of two users, i.e., $K=2$. The coordinate of the two users are randomly generated as (32.52,23.48,1.5) and (48.45,19.55,1.5), respectively.
\begin{figure}[!t]
	\centering	\includegraphics[width=.7\columnwidth]{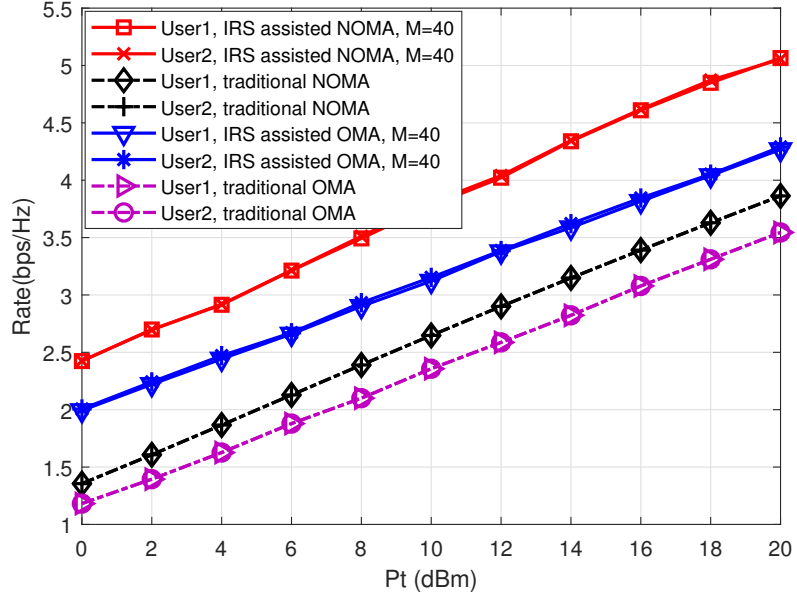}
	\caption{Rate comparison with different benchmarks for single-antenna BS case.} \label{fig:FigSim1}
\vspace{-0.4cm}
\end{figure}

First, Fig.~\ref{fig:FigSim1} plots the per-user rate performance versus the BS's transmission power $P$ for the proposed IRS-assisted NOMA and the three benchmarks, under the single-antenna BS setup (i.e., $N=1$). In general,  the proposed IRS-assisted NOMA achiepves significant rate gains compared to the benchmarks. Specifically, for the case of $P=10$ dBm, the proposed IRS-assisted NOMA improves the rate performance by 53.2\%, 38.5\%, and 14.3\%, compared to the benchmarks of traditional OMA without IRS, traditional NOMA without IRS, and IRS-assisted OMA, respectively. Compared to the traditional NOMA, the rate gain achieved by IRS-assisted NOMA, comes from the enhanced combined-channel strength and larger channel-strength differences introduced by the IRS. The additional rate gain of IRS-assisted NOMA compared to IRS-assisted OMA is due to the higher spectral efficiency of NOMA relative to OMA. Notice the superiority of NOMA compared to OMA still remains after introducing the IRS, as long as the difference of the combined channel strength is large enough. Also, the user 1 and user 2 achieve almost the same rate under each scheme, achieving best user fairness. The practical significance of this proposed IRS-assisted NOMA lies in that it enables the NOMA system to achieve higher rate and larger coverage than traditional NOMA and OMA systems while ensuring user fairness.

\begin{figure} [!t]
	\centering	\includegraphics[width=0.7\columnwidth]{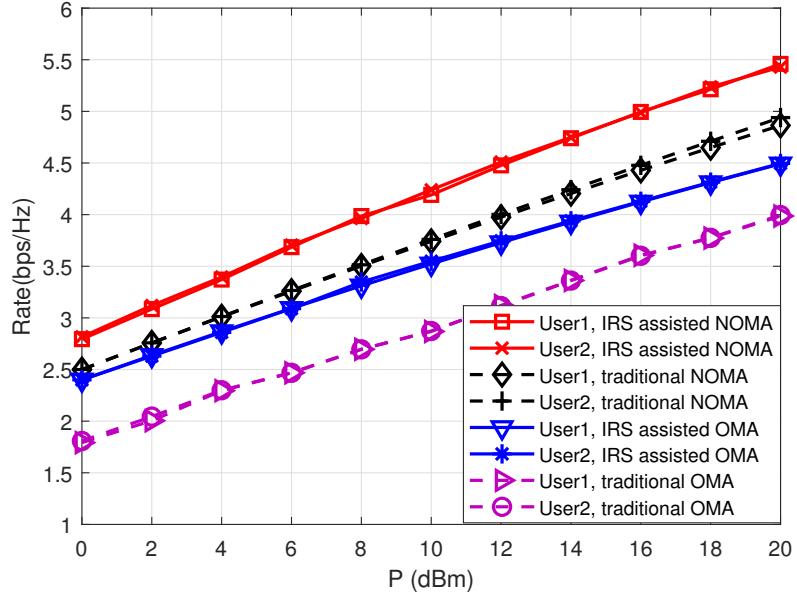}
	\caption{Rate comparison with different benchmarks for multi-antenna BS case.} \label{fig:FigSim2}
\vspace{-0.4cm}
\end{figure}


 Then, Fig.~\ref{fig:FigSim2} plots the per-user rate versus the BS's transmission power $P$ for the proposed IRS-assisted NOMA and the three benchmarks, under the multi-antenna BS setup. We set the number of BS antennas as $N=4$. Similar to the single-antenna setup, it is observed that the IRS-assisted NOMA achieves significant rate gains compared to the three benchmarks, which verifies the enhanced spectrum efficiency of NOMA and the benefits of the application of the IRS to the downlink  MISO-NOMA systems.


%

Moreover, Fig.~\ref{fig:FigSim3} compares the max-min rate performance of the proposed CCS-based user ordering scheme with that of the exhaustive search scheme. It is observed that the rate of the CCS-based user ordering scheme achieves almost the same performance as the exhaustive search scheme which needs to search all $K!$ possible user orderings, for both cases of a single-antenna BS and a multiple-antenna BS. This numerically verifies the performance advantage of the proposed CCS-based ordering scheme, besides it low-complexity feature.
\begin{figure}[!t]
	\centering	\includegraphics[width=0.7\columnwidth]{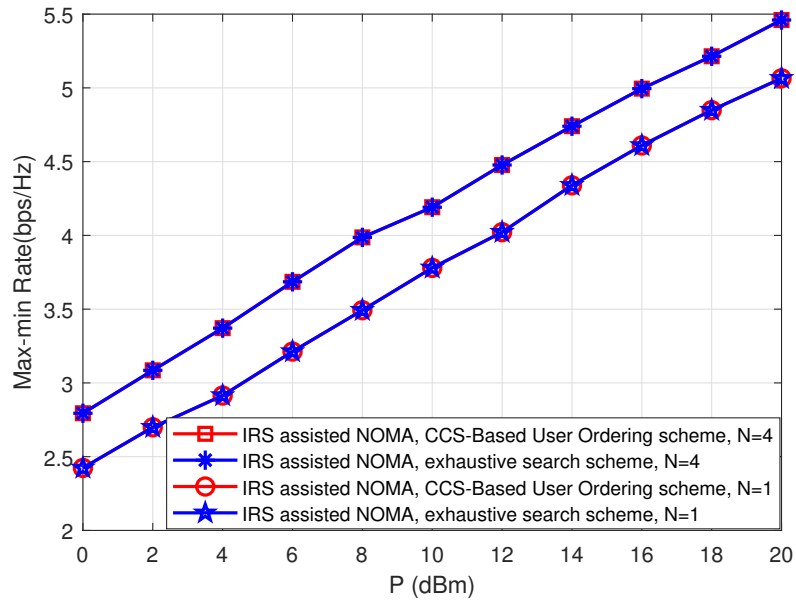}
	\caption{Max-min rate comparison with exhaustive-search user ordering scheme.}  \label{fig:FigSim3}
\vspace{-0.4cm}
\end{figure}
\begin{figure}[htbp]
	\centering \includegraphics[width=.7\columnwidth]{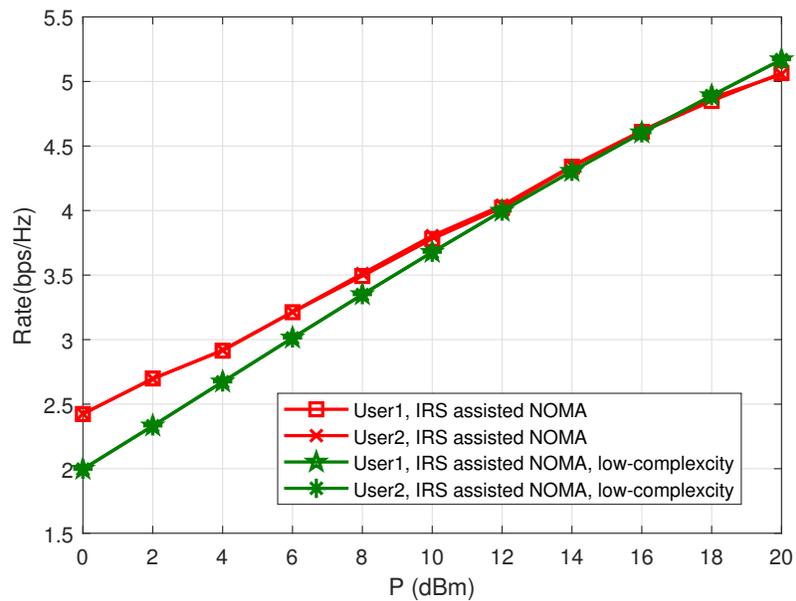}
	\caption{Max-min rate comparison with low-complexity solving scheme.}\label{fig:FigSim4} 
\vspace{-0.4cm}
\end{figure}

Also, Fig. \ref{fig:FigSim4} compares the performance of the low-complexity solution of IRS phase shifts $\Theta$ for two-user NOMA systems. When the transmission power $P$ at the BS, the low-complexity solution in closed form suffers from slight rate performance degradation compared to the solution achieved by the general algorithm. However, when $P$ is higher than 16 dBm, the low-complexity solution even outperforms the solution achieved by the general algorithm. This is because that for small or moderate $P$, the low-complexity solution maximizes the combined channel of the stronger user, without strengthening the weaker user's combined channel which may result into relatively low rate for the weaker user; while for large $P$, the low-complexity solution in closed-form is almost optimal, as proved in Proposition 1, outperforming than than the general solution. Therefore, for the two-user NOMA systems, the low-complexity solution is an efficient approach to determine the phase shift of the IRS, with reward of significant complexity reduction.


\subsection{Effects of IRS's Finite-Phase Resolution on Rate Performance}\label{resultforfinitephase}
In practical systems, the IRS structure has finite phase resolution and the implemented phase shifts depend on the number of quantization bits denoted as $B$. We numerically verify the effect of IRS's finite phase resolution on the rate performance. Each optimized continuous phase shift $\theta_m$ is quantized to its nearest discrete value in the set $\{0, \frac{2 \pi }{2^B}, \dots,\frac{2 \pi (2^B-1)}{2^B}\}$. Fig. \ref{fig:FigSim7} plots the max-min rate versus phase-quantization bits $B$'s under different transmission power for multi-antenna BS case. It is observed that the IRS's finite phase resolution in general degrades the max-min rate compared to the ideal case of IRS with infinite phase resolution, but the rate performance degradation becomes negligible very quickly as $B$ increases.


\begin{figure} [!t]
	\centering	\includegraphics[width=.7\columnwidth]{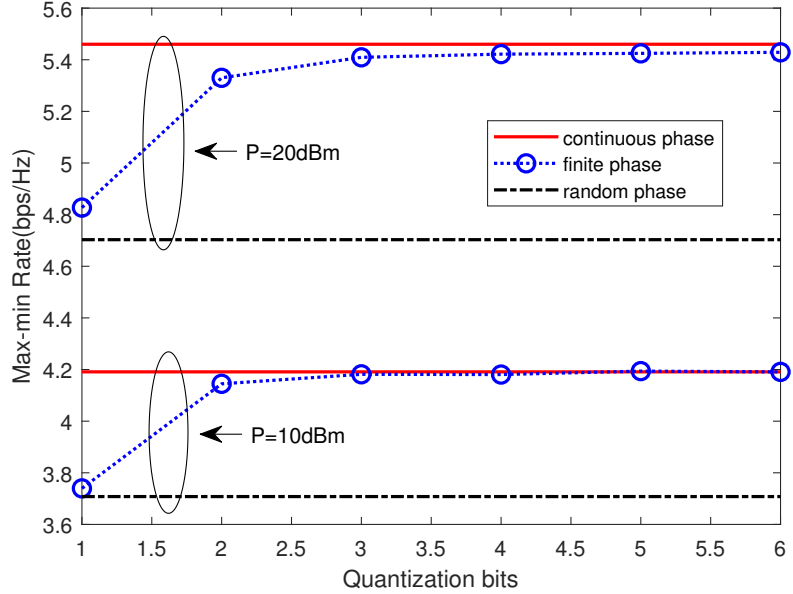}
	\caption{Max-min Rate comparison with different phase resolutions.}\label{fig:FigSim7}
\end{figure}


\subsection{Effect of Parameters $M$ and $K$ On Rate Performance}\label{resultforlowcomplexity}
In this subsection, the effects of main parameters on the rate performance are investigated. First, Fig.~\ref{fig:FigSim8} plots the max-min rate versus the number of reflecting elements $M$, for the transmission power of $0$dBm, $10$dBm and $20$dBm, respectively. As expected, for both transmission power, and the max-min rate increases with $M$. Then,

\begin{figure} [!t]
	\centering	\includegraphics[width=.7\columnwidth]{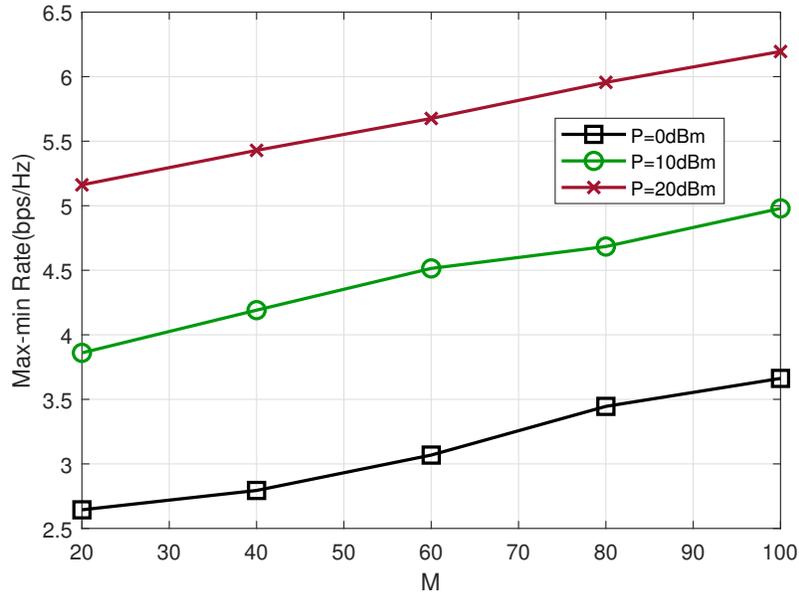}
	\caption{Max-min Rate versus number of reflecting elements $M$.}\label{fig:FigSim8} 
\vspace{-0.4cm}
\end{figure}

Fig.~\ref{fig:FigSim9} plots the rate performance versus the number of NOMA users $K$. It is observed from each curve in subfigure (a) that as $K$ increases, the sum rate increases first and then decreases, and achieves the maximum value for $K=8$. This reveals that there is a tradeoff between the number of NOMA users and the achievable sum rate. The higher sum rate can be obtained for larger number of reflecting elements $M$ or higher transmission power. From subfigure (b), we observe that for the two-user NOMA scenario, the users have almost the same rate, achieving the best fairness, and good rate fairness can always be guaranteed as the number of NOMA users $K$ increases. 



  \begin{figure}[h]
    \begin{subfigure}[b]{0.5\textwidth}
      \includegraphics[width=\textwidth]{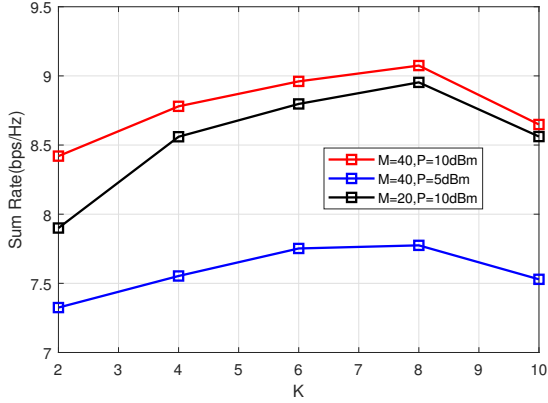}
      \caption{Sum Rate versus the number of users $K$}
    \end{subfigure}
    \begin{subfigure}[b]{0.5\textwidth}
      \includegraphics[width=\textwidth]{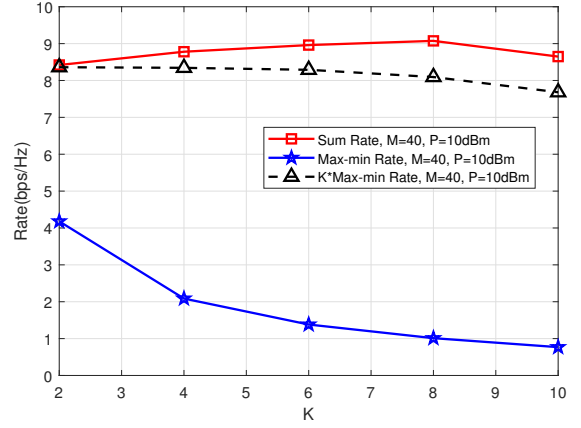}
      \caption{Max-min Rate versus the number of users $K$ }
    \end{subfigure}
    \caption{Rate performance versus number of users $K$}\label{fig:FigSim9}
  \end{figure}

\section{CONCLUSIONS}\label{Conclusions}\label{conclusion}
This paper has investigated the problem of rate optimization for an IRS-assisted downlink NOMA system. The minimum SINR (i.e., equivalently the rate) of all users are maximized by jointly optimizing the BS's transmit beamforming and the IRS's phase shifts. Efficient algorithms are proposed to find suboptimal solutions to the formulated non-convex problem, by leveraging the block coordinated decent and semidefinite relaxation techniques. Numerical results show that the IRS-assisted downlink NOMA system can enhance the rate performance significantly, compared to traditional NOMA without IRS and traditional orthogonal multiple access with/without IRS, and practical IRS with low phase resolution can approximate the best-achievable rate performance achieved by continuous phase shifts. Other interesting future work for IRS-assisted NOMA includes the outage performance analysis, rate performance under imperfect CSI, etc.
\begin{appendix}

\subsection{Proof of Proposition \ref{Proposition:Prop1}}\label{app:prop1}
\begin{proof}
Define the combined channel of the $k$-th user as $h_k\triangleq \mathbf{g}_{k}^{H}\mathbf{\Theta}\mathbf{f}+v_{k}^{H}$. Then we have
\begin{align}
\gamma_{1 \rightarrow 1}&=\frac{\alpha_1 P |h_1|^2}{\alpha_2 P |h_1|^2 +\sigma^{2} } \label{gamma11} \\
\gamma_{2 \rightarrow 2}&=\frac{\alpha_2 P |h_2|^2}{\sigma^{2} }. \label{gamma22}
\end{align}

For sufficiently high transmission power $P$ at the BS such that $\alpha_2 P |h_1|^2 \gg \sigma^{2}$, we have $\gamma_{1 \rightarrow 1} \approx \frac{\alpha_1}{\alpha_2}$. Given $\alpha_1$ and $\alpha_2$, $\gamma_{1 \rightarrow 1}$ is determined, while $\gamma_{2 \rightarrow 2}$ is monotonically increasing with respect to $|h_2|^2$. To maximize the minimum value between $\gamma_{1 \rightarrow 1}$ and $\gamma_{2 \rightarrow 2}$, it suffice to maximize the 2-nd user's combined channel strength $|h_2|^2=|\mathbf{g}_{2}^{H}\mathbf{\Theta}\mathbf{f}+v_{2}|^{2}$ by optimizing the phase shifts $\bTheta$.

Specifically, the combined channel strength $|h_2|^2$ can be rewritten as $|\be^{H} \diag(\mathbf{g}_{2}^{H})\mathbf{f} +v_{2}|^{2}$ by some variable substitutions, which can be further expressed as $\big|\sum_{i=1}^M \left( |[\mathbf{g}_{2}^{H}]_i|  |[\mathbf{f}]_{i}| e^{j(\theta_{i}+\varphi_{2,i}+\psi_{i})}\right)+ |v_{2}| e^{j\xi_2}\big|^2$. Hence, it is standard to show that the optimal phase shifts that maximize $|h_2|^2$ are given by $\theta_{i}=\xi_2-\varphi_{2,i}-\psi_{i}$, for $i=1,\ldots, M$.
\end{proof}

\subsection{Proof of Lemma \ref{lemma:Lem1}}\label{app:lemma1}
Let $\balpha^{*}$ denote the power allocation vector satisfying \eqref{eq:P2.4_equal} \eqref{eq:P2.4_sum} and $\gamma^{*}$ denote the obtained equal SINR. The optimality and uniqueness of $\balpha^{*}$ are proved in the sequel.

First, we prove $\balpha^{*}$ is the optimal solution to (P2.1). Assume $\balpha^{**}\neq \balpha^{*}$  with corresponding max-min SINR $\gamma^{**}$, and $\gamma^{**}>\gamma^{*}$. Due to the constraint \eqref{eq2:sumpowerallocationconstraint}, there must be an element of $\balpha^{**}$ is smaller than that of $\balpha^{*}$.  As proved in the next paragraph, if any element of $\balpha^{*}$ decrease, the max-min SINR would be smaller than $\gamma^{*}$, which contradicts with the previous assumption. Hence, $\balpha^{*}$ and $\gamma^{*}$ are the optimal solution and the optimal objective value, respectively.

For the $K$-th user, by using the fact that $\gamma_{K}$ is a monotonically increasing function of $\alpha_{K}$, the conclusion is clear. Therefore, in order to ensure $\gamma_{K} \geq \gamma^{*}$, $\alpha_{K}$ cannot be reduced. For the $K-1$ -th user, $\gamma_{K-1}=\frac{\alpha_{K-1} P|h_{K-1}|^{2}}{\alpha_{K} P|h_{K-1}|^{2}+\sigma^{2}}$, due to non-decreasing $\alpha_{K}$, the reduction of $\alpha_{K-1}$ will directly result into smaller $\gamma_{K-1}$ and thus $\alpha_{K-1}$ cannot be reduced either. The remaining $\gamma_{K-2}, \gamma_{K-3}, \ldots$, and $\gamma_{1}$ can be sequentially analyzed in the same manner.

Second, we prove that $\balpha^{*}$ is the unique solution to the equations \eqref{eq:P2.4_equal} and \eqref{eq:P2.4_sum}. From \eqref{eq:P2.4_equal}, we have the following recursive equations $\alpha_{K} (Q) =\frac{Q \sigma^2}{P|h_{k}|^{2}}$, and $\alpha_{k} (Q) =\frac{Q}{P|h_{k}|^{2}}\left({\sum \nolimits_{i=k+1}^{K}\alpha_{i} P|h_{k}|^{2}+\sigma^{2}}\right)$, for $k=K-1, K-2, \ldots, 1$. It can be easily shown that each $\alpha_{k} (Q)$ is strictly and monotonically increasing with $Q$. Therefore, $\sum \nolimits_{k=1}^{K}\alpha_{k}$ also strictly monotonically increases as $Q$ increases. Thus, there exist a unique positive value $Q^*$ which satisfies \eqref{eq:P2.4_equal} and \eqref{eq:P2.4_sum}. The optimal power allocation $\balpha^{*}$ are thus unique, and obtained as $\alpha_{k}^{*} (Q^*)$ for each $k$, which completes the proof.

\subsection{Proof of Theorem \ref{mythm:thm1}}{\label{app:theorem1}}
 Similarly to Subsection II-B in \cite{schubert2004solution}, the equations in \eqref{eq:P2.4_equal} can be written as
\begin{align}
\tilde{\balpha}\frac{1}{Q}=\bD \bPsi \tilde{\balpha}+\frac{\sigma^2}{P}\bD \textbf{1}.\label{eq:AppB1}
\end{align}
Multiplexing both sides by $\textbf{1}^T$ yields
\begin{align}
\frac{1}{Q}=\textbf{1}^T \bD \bPsi \tilde{\balpha}+\textbf{1}^T \frac{\sigma^2}{P}\bD \textbf{1}.\label{eq:AppB2}
\end{align}
 Define $ \tilde{\balpha}_{\sf ext}= [ \tilde{\balpha}; 1]$. From \eqref{eq:AppB1} and \eqref{eq:AppB2},
an eigenvalue system can be constructed as
 \begin{align}
\bPi \tilde{\balpha}_{\sf ext}= \frac{1}{Q} \tilde{\balpha}_{\sf ext},
\end{align}
where $Q$ is a reciprocal eigenvalue of the nonnegative matrix $\bPi$.

It is obvious that $Q > 0, \tilde{\balpha}_{\sf ext} \geq 0$ must be satisfied to represent physical meaning.
 According to Perron-Frobenius theory, for any nonnegative real matrix $\bB_{K\times K} \geq 0 $, whose spectral radius is $\rho(B)$, there exists a vector $\by \geq 0$ such that $\bB \by= \rho(B)\by$, thus the maximal eigenvalue $\rho(B)$ and the corresponding eigenvector are always nonnegative. Therefore, the optimal solution of the problem (P2.3) is given by
  \begin{align}
Q=\frac{1}{\lambda_{\max}(\bPi)}.
\end{align}
 And the optimal power allocation vector $\balpha$ is given by the first $K$ components of the dominant eigenvector of $\bPi$, which can be scaled such that its last component equals 1.

 \subsection{Proof of Theorem \ref{mythm:thmSDR}}{\label{app:theoremSDR}}
  Let $\mu_{t,k} \geq 0$ and $\nu \geq 0$ be the dual variables corresponding to the constraints given in \eqref{eq1.4:maxminthroughput}, \eqref{eq1.4:SICconstriant}, and \eqref{eq1.4:sumpowerallocationconstriant}, respectively, where $1\leq t \leq k \leq K$. Let $\bS_{k}\succeq 0$ be the dual variable corresponding constraints $\bOmega_{k} \succeq 0$ in \eqref{eq1.4:powerallocationconstriant}. The Lagrangian of (P1.4) is then written as
 \begin{align}
L(\{\bOmega_{k}\}, \mu_{t,k}, \nu, \{\bS_{k}\})=&-\sum_{t=1}^{K}\sum_{k=t}^{K}{\mu_{t,k} \left[\trace{(\bH_k \bOmega_t)}-Q\sum_{i=t+1}^{K}{\trace{(\bH_k \bOmega_i)}}-Q\sigma^2\right]}\nonumber\\
&+\nu \left[\sum_{k=1}^{K}{\trace{(\bOmega_k)}}-P\right]-\sum_{k=1}^{K}{\trace{(\bS_k\bOmega_k)}}.
\end{align}
Let $ \{\bOmega_{k}^{*}\}$, $\mu_{t,k}^{*}$, $\nu^{*}$ and $\{\bS_{k}^{*}\}$ be the optimal primal and dual variables, respectively. Since (P1.4) is convex for given $Q$ and satisfies the Slater¡¯s condition, the strong duality holds for this problem. As a result, the optimal primal and dual solutions should satisfy the Karush-Kuhn-Tucker (KKT) conditions given by
 \begin{align}
\nabla_{\bOmega_{l}} L(\{\bOmega_{k}^{*}\}, \mu_{t,k}^{*}, \nu^{*}, \{\bS_{k}^{*}\})=-\sum_{k=l}^{K}{\mu_{l,k}^{*}\bH_k}+Q\sum_{t=1}^{l-1}\sum_{k=t}^{K}{\mu_{t,k}^{*}\bH_k }+\nu^{*} \bI - \bS_{l}^{*}=0,\label{KKT1}\\
\bS_{l}^{*}\bOmega_{l}^{*}=0.\label{KKT2}
\end{align}
By multiplying \eqref{KKT1} by $\bOmega_{l}^{*}$ on both sides and substituting \eqref{KKT2} into the obtained equation, we have
 \begin{align}
\sum_{k=l}^{K}{\mu_{l,k}^{*}\bH_k}\bOmega_{l}^{*}-Q\sum_{t=1}^{l-1}\sum_{k=t}^{K}{\mu_{t,k}^{*}\bH_k }\bOmega_{l}^{*}=\nu^{*} \bOmega_{l}^{*}.
\end{align}
Recall $\bH_k=\bh_k^{H} \bh_k$ and $\bh_k= \mathbf{g}_{k}^{H}\mathbf{\Theta}\mathbf{F}+\bv_{k}^{H}=\be^{H}\diag\{\mathbf{g}_{k}^{H}\}\mathbf{F}+\bv_{k}^{H}$, by introducing
\begin{align}
\bGamma_k =\left[
\begin{array}{ccc}
\diag\{\mathbf{g}_{k}^{H}\}\mathbf{F} \\
\bv_{k}^{H}
\end{array}
\right],
\end{align}
the $\bh_k$ can be rewritten as $\bar{\be}^{H}\bGamma_k $.
Thus, we have
 \begin{align}
\left(\sum_{k=l}^{K}{\mu_{l,k}^{*}\bGamma_k^{H}\bar{\be}\bar{\be}^{H}\bGamma_k}-Q\sum_{t=1}^{l-1}\sum_{k=t}^{K}{\mu_{t,k}^{*}\bGamma_k^{H}\bar{\be}\bar{\be}^{H}\bGamma_k }\right)\bOmega_{l}^{*}=\nu^{*} \bOmega_{l}^{*}.
\end{align}
Since $\sum_{k=l}^{K}{\mu_{l,k}^{*}\bGamma_k^{H}\bar{\be}\bar{\be}^{H}\bGamma_k}=\bZ_k\bar{\be}\bar{\be}^{H}\bZ_k^{H}$, where $\bZ_k=\sum_{k=l}^{K}{\mu_{l,k}^{*}\bGamma_k^{H}}$, thus we have
\begin{align}
\rank \left(\sum_{k=l}^{K}{\mu_{l,k}^{*}\bGamma_k^{H}\bar{\be}\bar{\be}^{H}\bGamma_k}\right)\leq \rank \left(\bar{\be}\bar{\be}^{H}\right)=1.
 \end{align}
Similarly, $\rank \left(\sum_{t=1}^{l-1}\sum_{k=t}^{K}{\mu_{t,k}^{*}\bGamma_k^{H}\bar{\be}\bar{\be}^{H}\bGamma_k }\right)\leq1$.
The following derivations complete the proof
 \begin{align}
\rank({\bOmega_{l}^{*}})=&\rank{ \left(\sum_{k=l}^{K}{\mu_{l,k}^{*}\bGamma_k^{H}\bar{\be}\bar{\be}^{H}\bGamma_k}-Q\sum_{t=1}^{l-1}\sum_{k=t}^{K}{\mu_{t,k}^{*}\bGamma_k^{H}\bar{\be}\bar{\be}^{H}\bGamma_k }\right)\bOmega_{l}^{*}} \nonumber \\
&\leq\rank \left(\sum_{k=l}^{K}{\mu_{l,k}^{*}\bGamma_k^{H}\bar{\be}\bar{\be}^{H}\bGamma_k}\right)+\rank \left(\sum_{t=1}^{l-1}\sum_{k=t}^{K}{\mu_{t,k}^{*}\bGamma_k^{H}\bar{\be}\bar{\be}^{H}\bGamma_k }\right)
\leq2.
\end{align}

\end{appendix}
\renewcommand{\baselinestretch}{1.26}
\bibliographystyle{IEEEtran}
\bibliography{IEEEabrv,reference20191211NOMAIRS}

\end{document}